\documentclass[prd,reprint, nofootinbib,amsmath,amssymb, aps, floatfix]{revtex4-1}
\usepackage{dcolumn}% Align table columns on decimal point

\usepackage{bm}% bold math
\usepackage{amssymb,amsmath,amsthm,graphicx,ulem}
\usepackage[linktoc=page]{hyperref}

\usepackage{graphicx,subfigure}
\usepackage{epsfig}
\usepackage{amsmath}
\usepackage{amsfonts}
\usepackage{amssymb}
\usepackage[usenames]{color}
\usepackage{enumerate}

\usepackage[letterpaper,left=2.cm,right=2.cm,top=2.5cm,bottom=2.5cm]{geometry}
\usepackage[T1]{fontenc}
\newtheorem{thm}{Theorem}[section]

\newtheorem{cor}[thm]{Corollary} 
\newtheorem{lem}[thm]{Lemma}
\theoremstyle{remark}
\newtheorem{defn}[thm]{Definition}
\newtheorem{rem}[thm]{Remark}

\newtheorem{conv}[thm]{Convention}

\numberwithin{equation}{section}
\begin{document}
\normalem

\title{Proof of a New Area Law in General Relativity}
\author{Raphael Bousso}%
 \email{bousso@lbl.gov}
\affiliation{ Center for Theoretical Physics and Department of Physics\\
University of California, Berkeley, CA 94720, USA 
}%
\affiliation{Lawrence Berkeley National Laboratory, Berkeley, CA 94720, USA}

\author{Netta Engelhardt}
\email{engeln@physics.ucsb.edu}
\affiliation{Department of Physics, University of California, Santa Barbara, CA 93106, USA 
}%
\bibliographystyle{utcaps}

\begin{abstract}
A future holographic screen is a hypersurface of indefinite signature, foliated by marginally trapped surfaces with area $A(r)$. We prove that $A(r)$ grows strictly monotonically. Future holographic screens arise in gravitational collapse. Past holographic screens exist in our own universe; they obey an analogous area law. Both exist more broadly than event horizons or dynamical horizons. Working within classical General Relativity, we assume the null curvature condition and certain generiticity conditions. We establish several nontrivial intermediate results. If a surface $\sigma$ divides a Cauchy surface into two disjoint regions, then a null hypersurface $N$ that contains $\sigma$ splits the entire spacetime into two disjoint portions: the future-and-interior, $K^+$; and the past-and-exterior, $K^-$. If a family of surfaces $\sigma(r)$ foliate a hypersurface, while flowing everywhere to the past or exterior, then the future-and-interior $K^+(r)$ grows monotonically under inclusion. If the surfaces $\sigma(r)$ are marginally trapped, we prove that the evolution {\em must} be everywhere to the past or exterior, and the area theorem follows. A thermodynamic interpretation as a Second Law is suggested by the Bousso bound, which relates $A(r)$ to the entropy on the null slices $N(r)$ foliating the spacetime. In a companion letter, we summarize the proof and discuss further implications.
\end{abstract}

%\pacs{}
\maketitle

%%%%%%%%%%%%%%%%%%%%%%%%%%%%%%%%%%%%%%%%%%%%%%%%%%%%%%%%%%%%%%%%%%%%%%%%%%%
%%%%%%%%%%%%%%%%%%%%%%%%%%%%%%%%%%%%%%%%%%%%%%%%%%%%%%%%%%%%%%%%%%%%%%%%%%%

\tableofcontents

\section{Introduction}
\label{intro}

The celebrated laws of black hole thermodynamics~\cite{Bek72,Bek74,Haw74,Haw75} ascribe physical properties to the event horizon of a black hole. However, the event horizon is defined globally, as the boundary of the past of future infinity. Thus, the location of the thermodynamic object depends on the future history of the spacetime. For example, an observer in a perfectly flat spacetime region might already be inside a black hole, if a null shell is collapsing outside their past light-cone. By causality, Hawking radiation and the first and second law of black hole thermodynamics should have no manifestation for such an observer. Conversely, once a black hole has formed, its thermodynamic properties should be observable at finite distance, regardless of whether the collapsed region already coincides with the true event horizon, or is headed for substantial growth in the distant future.

Here we consider the problem of finding a geometric object that is locally defined, and which obeys a classical law analogous to one of the laws of thermodynamics. We will focus on the second law, whose manifestation in classical General Relativity is the statement that the area of certain surfaces cannot decrease. For the cross-sections of an event horizon this was proven by Hawking in 1971~\cite{Haw71}, but as noted above the event horizon is not locally defined.

We will formulate and prove a new area theorem. It is obeyed by what we shall call a {\em future (or past) holographic screen}, $H$. $H$ is a hypersurface foliated by marginally (anti-)trapped surfaces, which are called {\em leaves}. This definition is local, unlike that of an event horizon. It requires knowledge only of an infinitesimal neighborhood of each leaf.
\begin{figure*}[ht]
\subfigure[]{
\includegraphics[width=0.45 \textwidth]{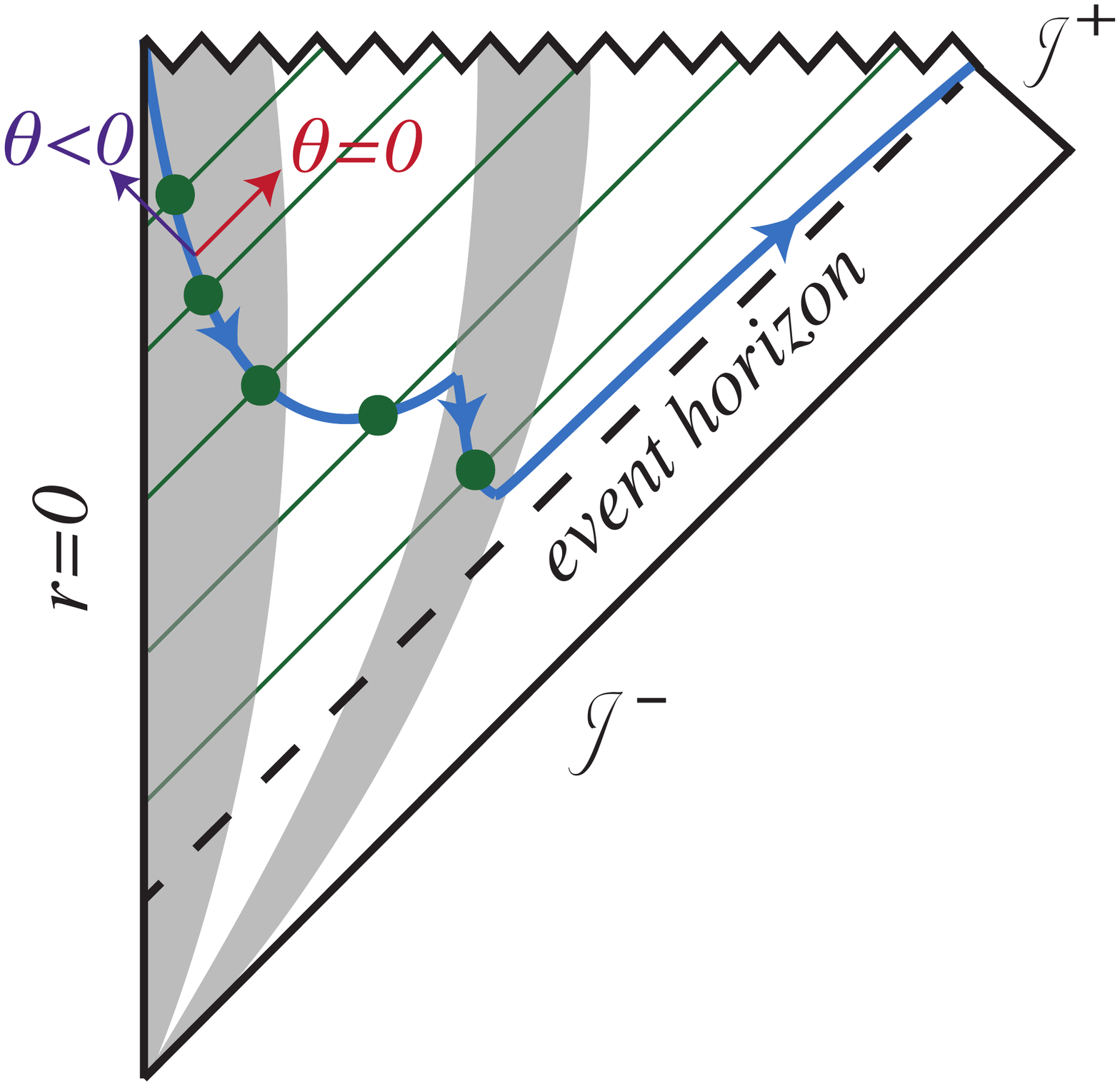}
%\label{fig-dustcollapse}
}
%\qquad
\subfigure[]{
 \includegraphics[width=0.45 \textwidth]{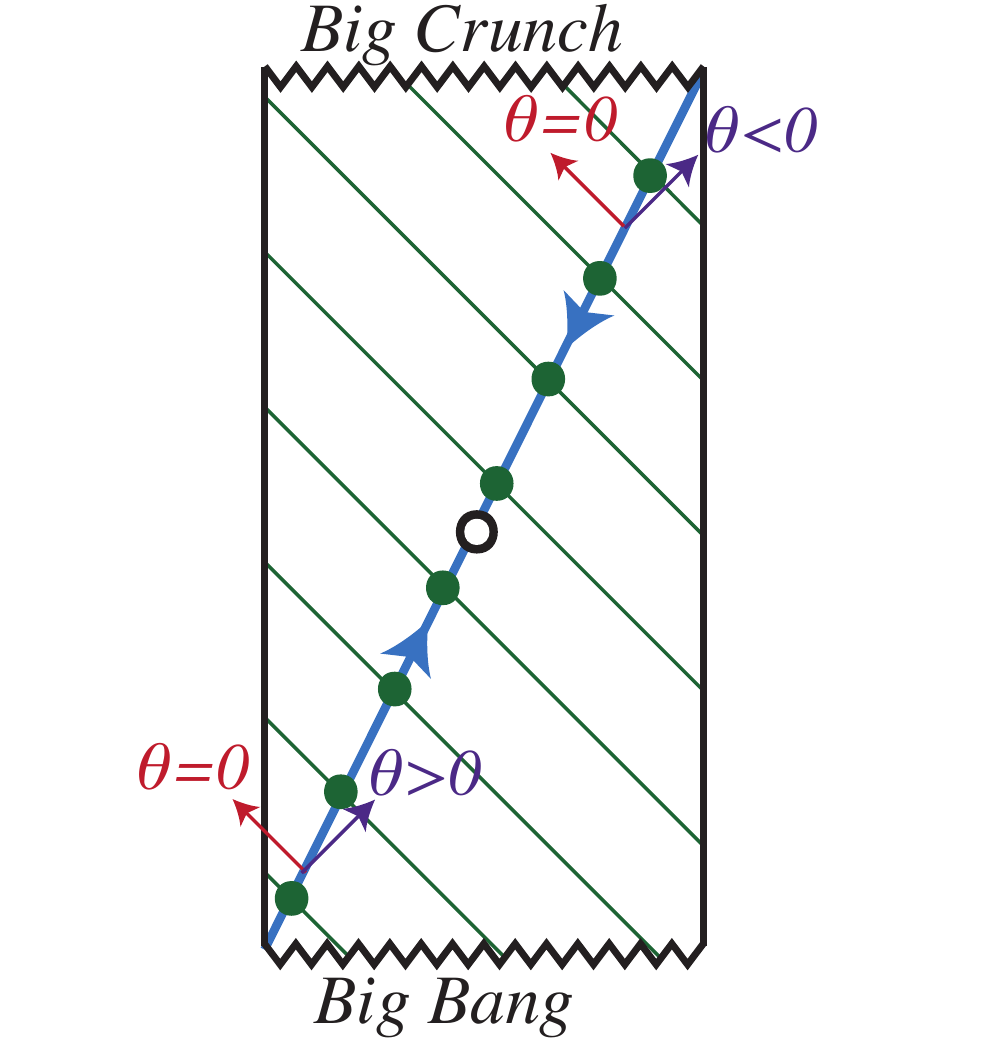}
%\label{fig-frw}
}
\caption{Penrose diagrams showing examples of holographic screens. The green diagonal lines show a null slicing of the spacetime; green dots mark the maximal area sphere on each slice. These surfaces combine to form a holographic screen (blue lines); we prove that their area increases monotonically in a uniform direction on the screen (blue triangles). (a) A black hole is formed by collapse of a star (inner shaded region); later another massive shell collapses onto the black hole (outer shaded region). At all other times an arbitrarily small amount of matter accretes (white regions); this suffices to satisfy our generic conditions. The black hole interior contains a future holographic screen that begins at the singularity and asymptotes to the event horizon. It is timelike in the dense regions and spacelike in the dilute regions.  (b) In a closed universe filled with dust, marginally antitrapped spheres form a past holographic screen in the expanding region; its area increases towards the future. Marginally trapped spheres form a future holographic screen in the collapsing region; its area increases towards the past. The equator of the three-sphere at the turnaround time (black circle) belongs to neither the past nor the future screen; it is extremal in the sense of Ref.~\cite{HubRan07}.}
\label{fig-examplesPRD}
\end{figure*} 
A future holographic screen exists (nonuniquely) in generic spacetimes that have a future event horizon. It is disjoint from the event horizon but it may asymptote to it; see Fig.~\ref{fig-examplesPRD}a. Past holographic screens exist in expanding universes such as ours, regardless of whether they have a past event horizon. Because $H$ is not defined in terms of distant regions, past and future holographic screens can exist in spacetimes with no distant boundary at all, such as a recollapsing closed universe; see Fig.~\ref{fig-examplesPRD}b. Our area law applies to all future and past holographic screens.

\noindent{\bf Relation to previous work} The notion of future or past holographic screen has roots in two distinct bodies of research, which had not been connected until now. It can be regarded as a refinement of the notion of ``preferred holographic screen hypersurface''~\cite{CEB2}, which need not have monotonic area. Alternatively, it can be viewed as a generalization of the notion of ``dynamical horizon'', which obeys a straightforward area law but is not known to exist in many realistic solutions. We will now discuss these two connections for context and attribution; see also~\cite{BouEng15a}. We stress, however, that our theorem and proof are self-contained. They rely only on classical General Relativity, and not, for example, on any conjecture about semiclassical or quantum gravity.

First, let us discuss the relation to the holographic principle. (See~\cite{Tho93,Sus95,FisSus98} for earlier work and Ref.~\cite{RMP} for a review.) To an arbitrary codimension 2 spatial surface $B$, one can associate a light-sheet~\cite{CEB1}: a null hypersurface orthogonal to $B$ with everywhere nonpositive expansion (i.e., locally nonincreasing area), in the direction away from $B$. The covariant entropy bound (Bousso bound)~\cite{CEB1} is the conjecture that the entropy of the matter on the light-sheet cannot exceed the area of $B$, in Planck units. The conjecture has broad support; it has been proven in certain limiting regimes~\cite{FlaMar99,BouFla03,StrTho03,BouCas14a,BouCas14b}.

There are four null directions orthogonal to any surface. In each direction, the orthogonal null congruence generates a null hypersurface with boundary $B$. The expansion in opposing directions, such as future-outward and past-inward, differs only by a sign. In typical settings, therefore, there will be two directions with initially negative expansion, each of which gives rise to a light-sheet. For example, a sphere in Minkowski space admits light-sheets in the future and past inward directions, but not in the outward directions. A large enough sphere near the big bang is anti-trapped: it admits light-sheets in the past inward and outward directions. Spheres near the singularity of a black hole are trapped: the light-sheets point in the future inward and outward directions.

However, it is possible to find surfaces that are {\em marginal}: they have vanishing expansion in one opposing pair of null directions. Hence they admit a pair of light-sheets whose union forms an entire null slice of the spacetime~\cite{CEB1}. In fact, in strongly gravitating regions one can readily construct a continuous family of marginal surfaces, which foliate a hypersurface called ``preferred holographic screen hypersurface''. The opposing pairs of light-sheets attached to each leaf foliate the spacetime. The Bousso bound is particularly powerful when applied to these light-sheets. It constrains the entropy of the entire spacetime, slice by slice, in terms of the area of the leaves. All quantum information in the spacetime can be stored on the leaves, at no more than about one qubit per Planck area. In this sense the world is a hologram.

For event horizons, a classical area theorem~\cite{Haw71} preceded the interpretation of area as physical entropy~\cite{Bek72}. For holographic screens, the present work belatedly supplies a classical area law for an object whose relevance to geometric entropy had long been conjectured~\cite{CEB2}. What took so long?

In fact, the notion of ``preferred holographic screen hypersurface'' lacked a key refinement, without which our theorem would not hold: the distinction between past and future holographic screens. The leaves of a ``preferred holographic screen hypersurface'' are {\em marginal}, that is, one orthogonal null congruence has vanishing expansion. However, they were not required to be either marginally trapped, or marginally anti-trapped. That is, no definite sign was imposed on the expansion of the second, independent orthogonal null congruence. Fig.~\ref{fig-examplesPRD}b shows a spacetime in which a ``preferred holographic screen hypersurface'' fails to obey an area law. Once we distinguish between marginally trapped and anti-trapped surfaces, however, we recognize that there are in fact two disconnected objects: a past and a future holographic screen. Each obeys an area law, as our proof guarantees, but in different directions of evolution. This is analogous to the distinction between past and future event horizons. From this perspective, it is not surprising that ``preferred holographic screen hypersurfaces'' fail to satisfy an area law without the refinement we introduce here.

This brings us to the second body of research to which the present work owes debt. Previous attempts to find a quasi-local alternative to the event horizon culminated in the elegant notions of a future outer trapping horizon (FOTH)~\cite{Hay93,Hay97,HayMuk98} or dynamical horizon~\cite{AshKri02,AshKri03,AshGal} (see~\cite{AshKri04,Booth05} for reviews). In a generic, classical setting their definitions are equivalent: a dynamical horizon is a spacelike hypersurface foliated by marginally trapped surfaces. 

``Preferred holographic screen hypersurface'' was a weaker notion than future holographic screen; ``dynamical horizon'' is a stronger notion. It adds not only the crucial refinement from marginal to marginally trapped, but also the requirement that the hypersurface be spacelike. This immediately implies that the area increases in the outward direction~\cite{Hay93,AshKri02}. (Note the brevity of the proof of Theorem~\ref{thm-area} below, which alone would imply an area law without the need for any of the previous theorems, if a spacelike assumption is imposed.) 

However, our present work shows that the spacelike requirement is not needed for an area theorem. This is important, because the spacelike requirement is forbiddingly restrictive~\cite{BoothBrits}: no dynamical horizons are known to exist in simple, realistic systems such as a collapsing star or an expanding universe dominated by matter, radiation, and/or vacuum energy. 

Thus, the notion of a dynamical horizon (or of a FOTH) appears to be inapplicable in a large class of realistic regions in which gravity dominates the dynamics. We are not aware of a proof of nonexistence. But we show here that an area theorem holds for the more general notion of future holographic screen, whose existence is obvious and whose construction is straightforward in the same settings. Thus we see little reason for retaining the additional restriction to hypersurfaces of spacelike signature, at least in the context of the second law.

In the early literature on FOTHs/dynamical horizons, future holographic screens were already defined and discussed, under the name ``marginally trapped tube''~\cite{AshKri04}.\footnote{The definition of ``trapping horizon''~\cite{Hay93} excludes the junctions between inner and outer trapping horizons and thus precludes the consideration of such objects as a single hypersurface.} Ultimately, two separate area laws were proven, one for the spacelike and one for the timelike portions of the future holographic screens. These follow readily from the definitions.  The first, for FOTHs/dynamical horizons, was mentioned above. The second states that the area decreases toward the future along any single timelike portion (known as ``future inner trapping horizons''~\cite{Hay93} or ``timelike membranes''~\cite{AshKri04}). 

In these pioneering works, no unified area law was proposed for ``marginally trapped tubes''/future holographic screens. Perhaps this is because it was natural to think of their timelike portions as future directed and thus area-decreasing. Moreover, the close relation to ``preferred holographic screen hypersurfaces''~\cite{CEB2} was not recognized, so the area of leaves lacked a natural interpretation in terms of entropy.\footnote{It is crucial that the entropy associated with the area of leaves on a future holographic screen $H$ is taken to reside on the light-sheets of the leaves, as we assert, and not on $H$ itself.  The latter choice---called a ``covariant bound'' in Refs.~\cite{He:2007qd, He:2007dh, He:2007xd, He:2008em} but related to~\cite{BakRey99} and distinct from~\cite{CEB1}---is excluded by a counterexample~\cite{KalLin99} and would not lead to a valid Generalized Second Law.} And finally, it is not immediately obvious that an area law can hold once timelike and spacelike portions are considered together. Indeed, the central difficulty in the proof we present here is our demonstration that such portions can only meet in ways that uphold area monotonicity for the entire future holographic screen under continuous flow. A key element of our proof builds on relatively recent work~\cite{Wal10QST}.

There is an intriguing shift of perspective in a brief remark in later work by Booth {\em et al.}~\cite{BoothBrits}. After explicitly finding a ``marginally trapped tube'' (i.e., what we call a future holographic screen) in a number of spherically symmetric collapse solutions, the authors point out that it could be considered as a single object, rather than a collection of dynamical horizon/``timelike membrane'' pairs. They note that with this viewpoint the area increases monotonically in the examples considered. Our present work proves that this behavior is indeed general.

Analogues of a first law of thermodynamics have been formulated for dynamical horizons and trapping horizons. We expect that this can be extended to future holographic screens. However, here we shall focus on the second law and its classical manifestation as an area theorem.

\noindent{\bf Outline \ }  In Sec.~\ref{sec-screens}, we give a precise definition of future and past holographic screens, and we establish notation and nomenclature. We also describe a crucial mathematical structure derived from the foliation of $H$ by marginally (anti-)trapped leaves $\sigma(r)$: there exists a vector field $h^a$ tangent to $H$ and normal to its leaves, which can be written as a linear combination of the orthogonal null vector fields $k^a$ and $l^a$. Its integral curves are called fibers of $H$. 

It is relatively easy to see that the area of leaves is monotonic if $h^a l_a$ has definite sign, i.e., if $H$ evolves towards the past or exterior of each leaf. The difficulty lies in showing that it does so everywhere.

Our proof is lengthy and involves nontrivial intermediate results. Given an arbitrary two-surface $\sigma$ that splits a Cauchy surface into complementary spatial regions, we show in Sec.~\ref{sec-csss} that a null hypersurface $N(\sigma)\supset \sigma$  partitions the entire spacetime into two complementary spacetime regions: $K^+(\sigma)$, the future-and-interior of $\sigma$; and $K^-(\sigma)$, the past-and-exterior of $\sigma$.

In Sec.~\ref{sec-ssmono}, we consider a hypersurface foliated by Cauchy-splitting surfaces $\sigma(r)$. We prove that $K^+(r)$ grows monotonically under inclusion, if the surfaces $\sigma(r)$ evolve towards their own past-and-exterior. This puts on a rigorous footing the equivalence (implicit in the constructions of~\cite{CEB2}) between foliations of $H$ and null foliations of spacetime regions. The proofs in Sec.~\ref{sec-monosplit} do not use all of the properties of $H$; in particular they do not use the marginally trapped property of its leaves. Thus our results up to this point apply to more general classes of hypersurfaces.

In Sec.~\ref{sec-arealaw}, we do use the assumption that the leaves of $H$ are marginally trapped, and we combine it with the monotonicity of $K^+(r)$ that we established for past-and-exterior evolution. This allows us to show that the evolution of leaves $\sigma(r)$ on a future holographic screen $H$ must be {\em everywhere} to the past or exterior (assuming the null energy condition and certain generic conditions). This is the core of our proof. We then demonstrate that such evolution implies that the area $A(r)$ of $\sigma(r)$ increases strictly monotonically with $r$. 

We close Sec.~\ref{sec-arealaw} with a theorem establishing the uniqueness of the foliation of a given holographic screen. The holographic screens themselves are highly nonunique. For example, one can associate a past (future) holographic screen with any observer, by finding the maximal area surfaces on the past (future) light-cones of each point on the observer's worldline.

\section{Holographic Screens}
\label{sec-screens}

We assume throughout this paper that the spacetime is globally hyperbolic (with an appropriate generalization for asymptotically AdS geometries~\cite{Wal10QST,EngWal14}). We assume the null curvature condition (NCC): $R_{ab}k^{a}k^{b}\geq 0$ where $k^{a}$ is any null vector. In a spacetime with matter satisfying Einstein's equations this is equivalent to the null energy condition: $T_{ab}k^{a}k^{b}\geq 0$.

\begin{defn}
A {\em future holographic screen}~\cite{CEB2} (or {\em marginally trapped tube}~\cite{AshGal, AshKri04}) $H$ is a smooth
%\footnote{Strictly, our proof only requires that $H$ is at least $C^2$, so that $\alpha$ (defined below) is continuous. We expect that the proof could be further strengthened to require only that the leaves be $C^2$ and $\alpha$ be piecewise continuous.} 
hypersurface admitting a foliation by marginally trapped surfaces called {\em leaves}.

A {\em past holographic screen} is defined similarly but in terms of marginally anti-trapped surfaces. Without loss of generality, we will consider future holographic screens in general discussions and proofs.

By {\em foliation} we mean that every point $p\in H$ lies on exactly one leaf.  A {\em marginally trapped surface} is a codimension 2 compact spatial surface $\sigma$ whose two future-directed orthogonal null geodesic congruences satisfy
\begin{eqnarray}
\theta_{k} &=& 0 \label{marginal}~,\\
\theta_{l} &<&  0 \label{trapped}~.
\end{eqnarray} 
The opposite inequality defines ``marginally anti-trapped'', and thus, past holographic screens. Here $\theta_{k}= \hat{\nabla}_a k^a$ and $\theta_{l} = \hat{\nabla}_a l^a$ are the null expansions~\cite{Wald} (where $\hat{\nabla}_{a}$ is computed with respect to the induced metric on $\sigma$), and $k^{a}$ and $l^{a}$ are the two future directed null vector fields orthogonal to $\sigma$.

We will refer to the $k^a$ direction as {\em outward} and to the $l^a$ direction as {\em inward}. For screens in asymptotically flat or AdS spacetimes, these notions agree with the intuitive ones. Furthermore, in such spacetimes any marginally trapped surface, and hence any holographic screen, lies behind an event horizon. However, holographic screens may exist in cosmological spacetimes where an independent notion of outward, such as conformal infinity, need not exist (e.g., a closed FRW universe).  In this case the definition of $H$ requires only that there exist some continuous assignment of $k^a$ and $l^a$ on $H$ such that all leaves are marginally trapped. See Fig.~\ref{fig-examplesPRD} for examples of holographic screens.
\end{defn}

\begin{defn}
The defining foliation of $H$ into leaves $\sigma$ determines a $(D-2)$-parameter family of leaf-orthogonal curves $\gamma$, such that every point $p\in H$ lies on exactly one curve that is orthogonal to $\sigma(p)$. We will refer to this set of curves as the {\em fibration of $H$}, and to any element as a {\em fiber} of $H$.
\end{defn}

\begin{conv} \label{conv-rh}
Thus it is possible to choose a (non-unique) {\em evolution parameter} $r$ along the screen $H$ such that $r$ is constant on any leaf and increases monotonically along the fibers $\gamma$. We will label leaves by this parameter: $\sigma(r)$. 

The tangent vectors to the fibers define a vector field $h^a$ on $H$. For any choice of evolution parameter the normalization of this vector field can be fixed by requiring that the function $r$ increases at unit rate along $h^a$: $h(r)=h^a\, (dr)_a=1$. (Since $H$ can change signature, unit normalization of $h^a$ would be possible only piecewise, and hence would not be compatible with the desired smoothness of $h^a$.)
\end{conv}

\begin{rem}
Since fibers are orthogonal to leaves, a tangent vector field $h^a$ can be written as a (unique) linear combination of the two null vector fields orthogonal to each leaf:
\begin{equation}
h^{a} = \alpha l^{a} + \beta k^{a}
\label{eq-hlk}
\end{equation}
Moreover, the foliation structure guarantees that $h^a$ vanishes nowhere: it is impossible to have $\alpha=\beta=0$ anywhere on $H$. (These remarks hold independently of the requirement that each leaf be marginally trapped.) \label{rem-h}
\end{rem}
\begin{figure}[ht]

\includegraphics[width=2in]{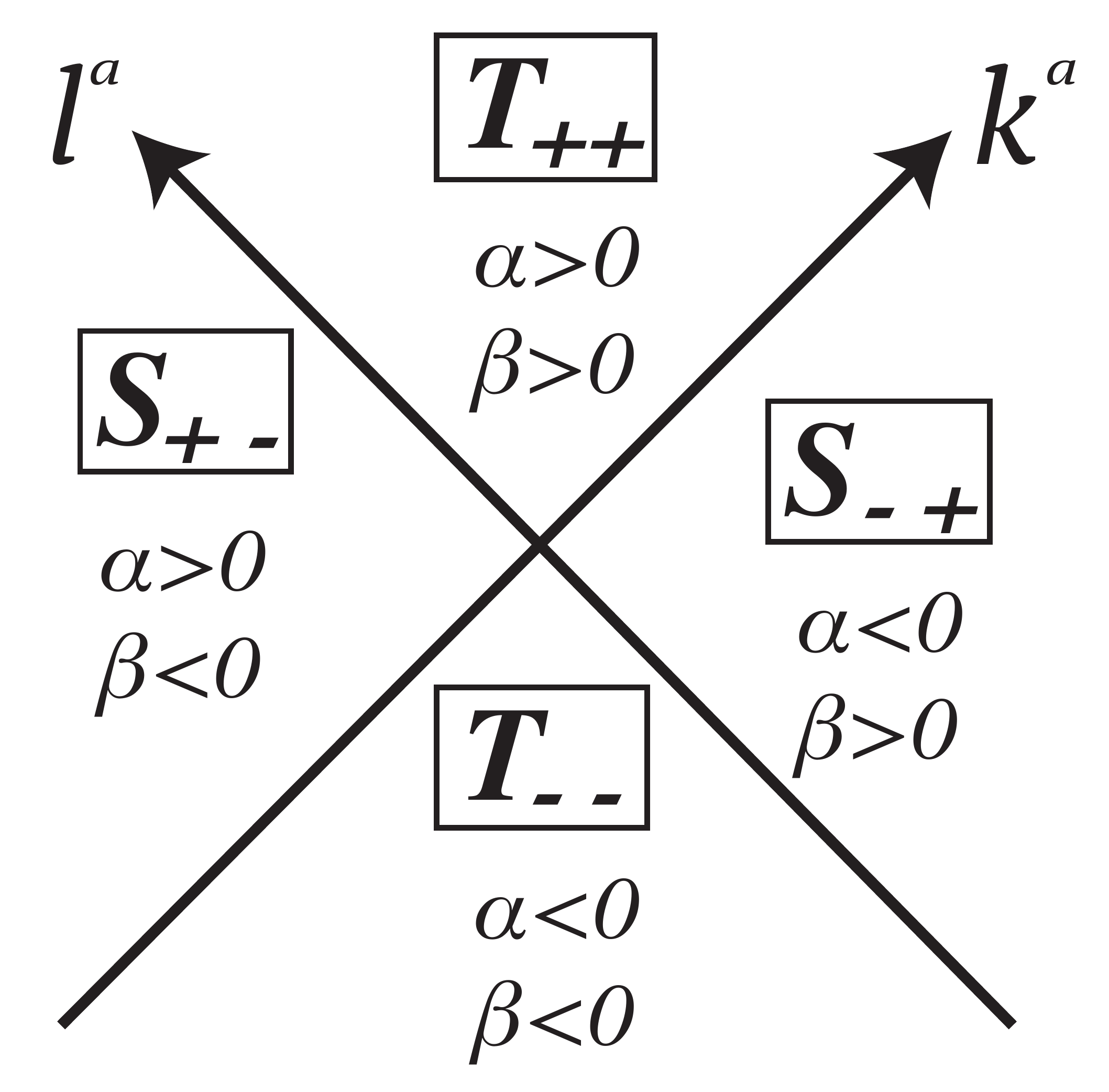}
\caption{The null vectors $l^{a}$ and $k^{a}$ orthogonal to a leaf $\sigma$ of the foliation of $H$ at some point. The evolution of $H$ is characterized by vector $h^a$ normal to the leaves and tangent to $H$. Depending on the quadrant $h^a$ points to, $H$ evolves locally to the future, exterior, past, or interior (clockwise from top).}
\label{fig-st}

\end{figure}

\begin{conv} \label{conv-stnotation}
As shown in Fig.~\ref{fig-st}, $h^a$ is spacelike and outward-directed if $\alpha<0, \beta>0$; timelike and past-directed if $\alpha<0, \beta<0$; spacelike and inward-directed if $\alpha>0, \beta<0$; and finally, timelike and future-directed if $\alpha>0, \beta>0$. We denote such regions, in this order (and somewhat redundantly): $S_{-+},T_{--},S_{+-},T_{++}$. 
\end{conv}
\begin{rem} \label{rem-result}
Our key technical result below will be to demonstrate that $\alpha$ cannot change sign on $H$. Thus on a given screen $H$, either only the first two, or only the second two possibilities are realized. (The latter case can be reduced to the former by taking $r\to -r$.)
\end{rem}
\begin{rem} \label{rem-borders}
Because $\alpha$ and $\beta$ cannot simultaneously vanish, $S_{+-}$ and $S_{--}$ regions cannot share a boundary or be separated by a null region; they must be separated by a timelike region. Similarly $T_{++}$ and $T_{--}$ regions must be separated by a spacelike region.
\end{rem}

Below we will consider only holographic screens that satisfy additional technical assumptions:
\begin{defn} \label{def-technical}
A holographic screen $H$ is {\em regular} if
\begin{enumerate}[(a)] % (a), (b), (c), ...
\item \label{def-technical1} the {\em first generic condition} is met, that $R_{ab} k^a k^b + \varsigma_{ab}\varsigma^{ab}>0$ everywhere on $H$, where $\varsigma_{ab}$ is the shear of the null congruence in the $k^a$-direction;
\item \label{def-technical2} the {\em second generic condition} is met: let $H_+$, $H_-$, $H_0$ be the set of points in $H$ with, respectively, $\alpha>0$, $\alpha<0$, and $\alpha=0$. Then $H_0= \dot H_- = \dot H_+$. 
\item \label{def-technical3} every inextendible portion $H_i\subset H$ with definite sign of $\alpha$ either contains a complete leaf, or is entirely timelike.
\item \label{def-technical4} every leaf $\sigma$ splits a Cauchy surface $\Sigma$ into two disjoint portions $\Sigma^\pm$.
\end{enumerate} 
\end{defn}
Analogous assumptions have been used in the more restricted context of dynamical horizons. The first generic condition is identical to the regularity condition of~\cite{AshGal}. Together with the null curvature condition, $R_{ab}k^a k^b\geq 0$, it ensures that the expansion of the $k^a$-congruence becomes negative away from each leaf. The second generic condition excludes the degenerate case where $\alpha$ vanishes along $H$ without changing sign. Either condition excludes the existence of an open neighborhood in $H$ with $\alpha=0$. Both are aptly called ``generic'' since they can fail only in situations of infinitely fine-tuned geometric symmetry and matter distributions. % (Strictly, we require that the first assumption hold only where $H$ is timelike; however, as stated it is already quite weak.)
The third assumption is substantially weaker than the definition of a dynamical horizon, since we do not require global spacelike signature of $H$. The fourth assumption will play a role analogous to the assumption of achronality of the dynamical horizon. It holds in typical spacetimes of interest (including settings with nontrivial spatial topology, such as $S^1\times S^2$, as long as the holographic screen is sufficiently localized on the sphere). We leave the question of relaxing some or all of these assumptions to future work.

\begin{rem} \label{rem-s0exists}
Assumption~\ref{def-technical}.c and Remark~\ref{rem-borders} imply that $H$ contains at least one complete leaf with definite sign of $\alpha$.
\end{rem}
\begin{conv}
\label{conv-orient}
Let $\sigma(0)\subset H$ be an arbitrary leaf with definite sign of $\alpha$. We will take the parameter $r$ to be oriented so that $\alpha< 0$ on $\sigma(0)$, and we take $r=0$ on $\sigma(0)$. By convention \ref{conv-rh} this also determines the global orientation of the vector field $h^a$. For past holographic screens, it is convenient to choose the opposite convention, $\alpha>0$ on $\sigma(0)$.
\end{conv}

\section{Leaves Induce a Monotonic Spacetime Splitting}
\label{sec-monosplit}

In this section, we will use only a subset of the defining properties of a holographic screen. In Sec.~\ref{sec-csss}, we examine the implications of Assumption~\ref{def-technical}.d, that each leaf split a Cauchy surface. We show that a null surface orthogonal to such a leaf splits the entire spacetime into two disconnected regions  $K^\pm(\sigma)$.

In Sec.~\ref{sec-ssmono}, we use the foliation property of the holographic screen. (However, nowhere in this section do we use the condition that each leaf be marginally trapped, or Assumptions~\ref{def-technical}.a-c.) We show that in portions of $H$ where $\alpha$ is of constant sign, the sets $K^\pm(\sigma(r))$ satisfy inclusion relations that are monotonic in the evolution parameter $r$. 

Together these results imply that an $\alpha<0$ foliation of any hypersurface $H$ into Cauchy-splitting surfaces $\sigma$ induces a null foliation of the spacetime, such that each null hypersurface $N(\sigma)$ splits the entire spacetime into disconnected regions $K^\pm(\sigma)$. 

In the following section, we will add the marginally trapped condition and the remaining technical assumptions, to show that on a future holographic screen, $\alpha$ {\em must}\/ have constant sign.

\subsection{From Cauchy Splitting to Spacetime Splitting}
\label{sec-csss}

By Assumption~\ref{def-technical}.d, every leaf $\sigma$ splits a Cauchy surface $\Sigma$ into two disconnected portions $\Sigma^+$ and $\Sigma^-$:
\begin{equation}
\Sigma = \Sigma^+ \cup \sigma \cup \Sigma^-~,~~\sigma = \dot\Sigma^\pm~.
\end{equation}
We take $\Sigma^\pm$ to be open in the induced topology on $\Sigma$, so that $\Sigma^\pm\cap\sigma=\varnothing$.
We consider the following sets shown in Fig.~\ref{fig-sets}a:
\begin{itemize} 
\item $I^+(\Sigma^+)$, the chronological future of $\Sigma^+$: this is the set of points that lie on a timelike future-directed curve starting at $\Sigma^+$. (Note that this set does not include $\Sigma^+$.)
\item $D^-(\Sigma^+)$, the past domain of dependence of $\Sigma^+$: this is the set of points $p$ such that every future-directed causal curve through $p$ must intersect $\Sigma^+$. (This set does include $\Sigma^+$.)
\item Similarly, we consider $I^-(\Sigma^-)$ and $D^+(\Sigma^-)$.
\end{itemize} 

\begin{figure*}[ht]

\subfigure[ ]{

\includegraphics[height=0.4 \textwidth]{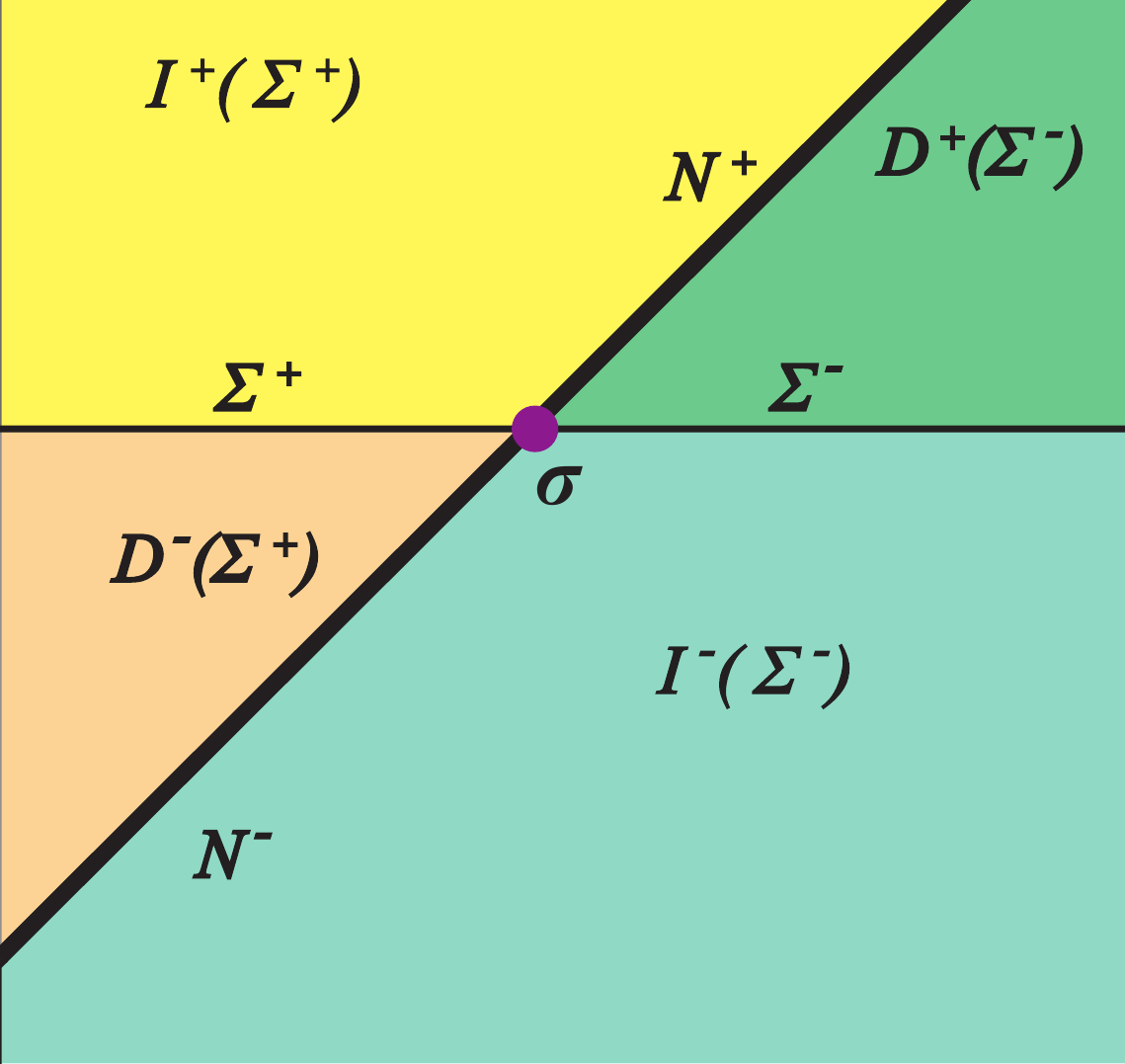}
\label{fig-sets-a}}
\qquad
\subfigure[]{
\includegraphics[height=0.4\textwidth]{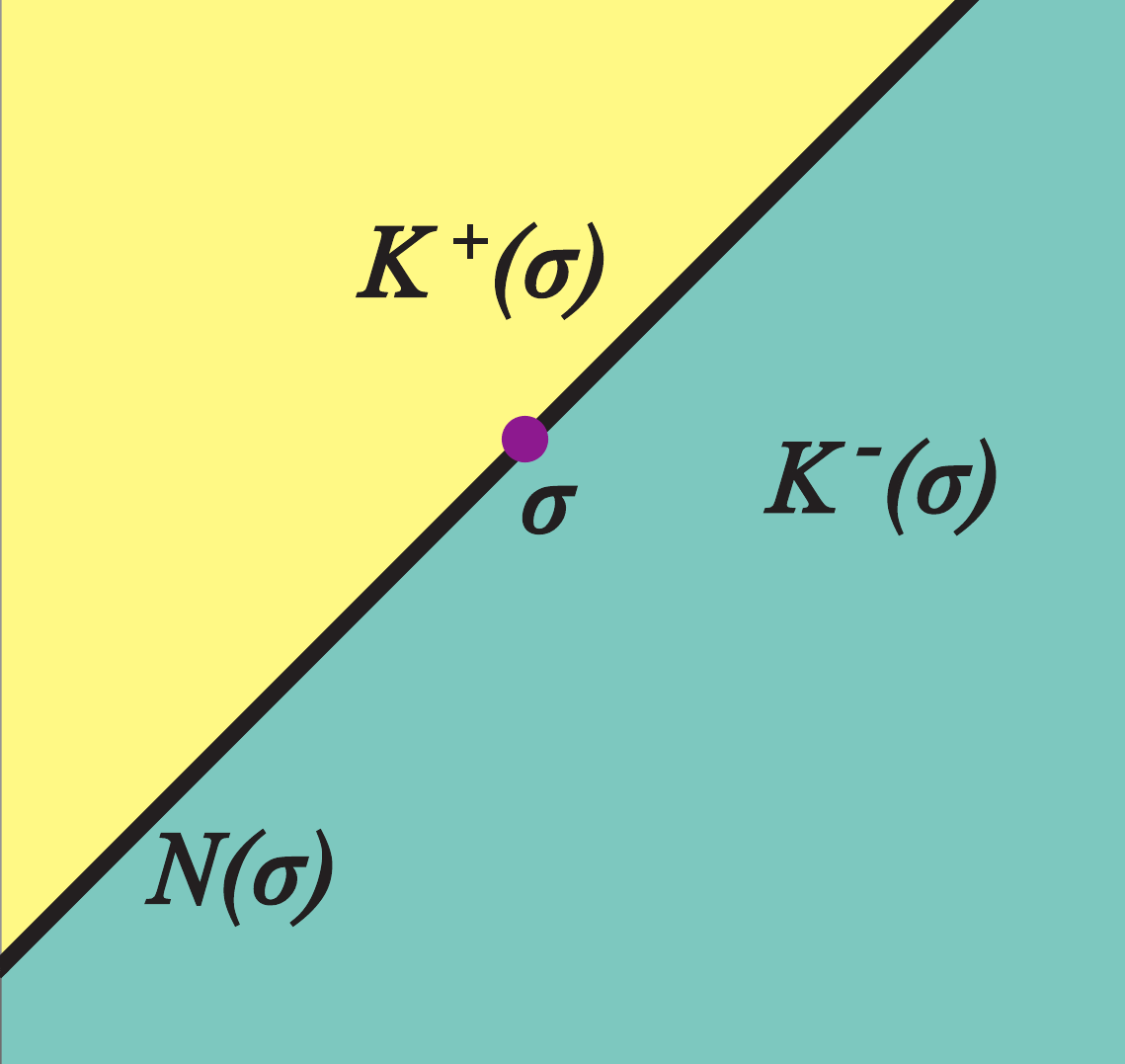}
\label{fig-sets-b}}
\caption{(a) Each leaf $\sigma$ splits a Cauchy surface. This defines a partition of the entire spacetime into four regions, given by the past or future domains of dependence and the chronological future or past of the two partial Cauchy surfaces. (b) The pairwise unions $K^\pm$ depend only on $\sigma$, not on the choice of Cauchy surface. They can be thought as past and future in a null foliation defined by the lightsheets $N$.}
\label{fig-sets}
\end{figure*}

\begin{defn}\label{def-sets}
From the Cauchy-splitting property of $\sigma$, it follows\footnote{The proofs of the following statements are straightforward and use only well-known properties of $I^\pm$ and $D^\pm$.} that the four sets defined above have no mutual overlap. However they share null boundaries:
\begin{align} 
N^+(\sigma) & \equiv & \dot I^+(\Sigma^+) -\Sigma^+ = \dot D^+(\Sigma^-) - I^-(D^+(\Sigma^-)) \\
N^-(\sigma) & \equiv & \dot I^-(\Sigma^-) - \Sigma^- = \dot D^-(\Sigma^+) - I^+(D^-(\Sigma^+))
\end{align} 
Note that $N^+(\sigma)\cap N^-(\sigma)=\sigma$. We define
\begin{eqnarray} 
K^+(\sigma) & \equiv & I^+(\Sigma^+) \cup D^-(\Sigma^+)-N^+(\sigma)~; \\
K^-(\sigma) & \equiv & D^+(\Sigma^-) \cup I^-(\Sigma^-)-N^-(\sigma)~; \\
N(\sigma) &\equiv & N^+(\sigma) \cup N^-(\sigma)
\end{eqnarray}
Thus
\begin{equation}
N(\sigma) = \dot K^+(\sigma) = \dot K^-(\sigma)~;
\end{equation}
and the sets $N$, $K^+$, and $K^-$ provide a partition of the spacetime (Fig.~\ref{fig-sets}b). 
\end{defn}

\begin{lem}\label{lem-secondN}
There exists an independent characterization of $N^+$, $N^-$, and thus of $N$:
$N^+(\sigma)$ is generated by the future-directed null geodesic congruence orthogonal to $\sigma$ in the $\Sigma^-$ direction up to intersections: $p\in N^+(\sigma)$ if and only if no conjugate point or nonlocal intersection with any other geodesic in the congruence lies between $\sigma$ and $p$.
\end{lem}
This follows from a significantly strengthened version of Theorem 9.3.11 in Ref.~\cite{Wald}, a proof of which will appear elsewhere. Similarly $N^-$ is generated by the past-directed $\sigma$-orthogonal null congruence towards $\Sigma^+$. (Hence if $\sigma$ is marginally trapped then $N^\pm$ both are light-sheets of $\sigma$~\cite{CEB1}.)

\begin{cor}\label{cor-onlysigma}
Lemma~\ref{lem-secondN} implies that $N$ depends only on $\sigma$, not on the Cauchy surface $\Sigma$.  Moreover, the sets $K^+$ and $K^-$ are then uniquely fixed by the fact that $N$ splits the spacetime: $K^+$ is the largest connected set that contains $I^+(N)$ but not $N$.
\end{cor}
Thus our use of $\sigma$ (as opposed to $\Sigma^+$ and/or $\Sigma^-$) as the argument of the sets $K^\pm$, $N^\pm$ is appropriate. 

\subsection{Monotonicity of the Spacetime Splitting}
\label{sec-ssmono}

Until now, we have only used the Cauchy-splitting property of $\sigma$. We will now consider a family of such leaves, $\sigma(r)$, that foliate a hypersurface ${\cal H}$. (We use this notation instead of $H$, in order to emphasize that ${\cal H}$ need not satisfy the additional assumptions defining a future holographic screen.) A tangent vector field $h^a$ can be defined as described in Remark~\ref{rem-h}, with decomposition $h^a = \alpha l^a + \beta k^a$ into the null vectors orthogonal to each leaf. We take $\Sigma^+$ to be the side towards which the vector $l^a$ points. (This convention anticipates Sec.~\ref{sec-arealaw}. In the current section, $k^a$ and $l^a$ need not be distinguished by conditions on the corresponding expansions.) To simplify notation, we denote $K^+(\sigma(r))$ as $K^+(r)$, etc.

\begin{thm}\label{thm-kmono}
Consider a foliated hypersurface ${\cal H}$ with tangent vector field $h^a$ defined as above. Suppose that $\alpha<0$ on all leaves $\sigma(r)$ in some open interval, $r\in I$. Then
\begin{equation}
\bar K^+(r_1)\subset K^+(r_2)~,
\label{eq-kmono}
\end{equation}
or equivalently $K^-(r_1)\supset \bar K^-(r_2)$, for all $r_1, r_2 \in I$ with $r_1<r_2$. That is, the sets $K^\pm(r)$ are monotonic in $r$ under inclusion; and the monotonicity is strict in the strong sense that the entire boundary  $N(r_1)$ of the open set $K^+(r_1)$ is contained in the open set $K^+(r_2)$.\footnote{It is not difficult to strengthen this theorem by proving the converse. However this requires using assumption \ref{def-technical}.b which is used nowhere else in this Section. Moreover, the converse is not needed in this paper.}
\end{thm}
\begin{proof}
% ``If:'' We prove the contrapositive. Consider an open interval $I$ and suppose that there exists a leaf $\sigma(r)$, $r\in I$, that contains a point with $\alpha\geq 0$. By assumption \ref{def-technical}.b, there exists a nearby point with $\alpha>0$, and by continuity, there will exist an open neighborhood $O\subset {\cal H}$ with $\alpha>0$. We choose leaves $\sigma(r_1)$, $\sigma(r_2)$ that both intersect $O$, and which are sufficiently close so that at least one fiber $\gamma$ (integral curve of $h^a$) lies entirely in $O$ between $r_1$ and $r_2$. Consider a point $p\in \sigma(r_2)\cap \gamma$. Because $K^+(r_2)$ is open and $p$ lies on its boundary, we have $p\notin K^+(r_2)$. Because $\alpha>0$, the curve element $\gamma$ from $\sigma(r_1)$ to $p$ is future or inward directed; for sufficiently small $r_2-r_1$ this implies $p\in K^+(r_1)$. Hence, Eq.~(\ref{eq-kmono}) is violated for some choice of $r_1, r_2\in I$.
%``Only if:'' 
We will first prove the inclusion monotonicity of $K^\pm$ under an infinitesimal evolution step $r\to r+\delta r$. The assumption that $\alpha<0$ implies that ${\cal H}$ locally evolves towards the past or exterior of its leaves: for sufficiently small $\delta r<0$,
\begin{equation}
\sigma(r+\delta r)\subset K^-(r)~.
\label{eq-rdr}
\end{equation}
Since $K^-(r)\cap K^+(r)=\varnothing$, it follows that ${\cal H}$ cannot locally evolve into the future or interior of any of its leaves:
\begin{equation}
\sigma(r+\delta r)\cap K^+(r)=\varnothing~.
\end{equation}
Let $\delta {\cal H}$ be the small portion of ${\cal H}$ between $r$ and $r+\delta r$; the above results imply that
\begin{equation}
\delta {\cal H} \subset K^-(r)~,~~~\delta {\cal H} \cap K^+(r)=\varnothing~.
\end{equation}

By Corollary~\ref{cor-onlysigma}, we may choose the sets $\Sigma^+(r)$ to suit our convenience. 
It is instructive to consider first the special case where we can find a Cauchy surface such that $\Sigma^+(r+\delta r) = X$, where
\begin{equation}
X\equiv \bar\Sigma^+(r)\cup \delta {\cal H}~,
\label{eq-x}
\end{equation}
and we recall that an overbar denotes the closure of a set.  This is the case when $\beta>0$ between $\sigma(r)$ and $\sigma(r+\delta r)$, i.e., if $\delta {\cal H}$ is spacelike. Both the future of a set, and the past domain of dependence of a set cannot become smaller when the set is enlarged; hence,
\begin{eqnarray} 
I^+(X) & \supset & I^+(\Sigma^+(r))~,\nonumber \\
D^-(X) & \supset & D^-(\Sigma^+(r))~,
\label{eq-idsuper}
\end{eqnarray} 
and so the infinitesimal version of Eq.~(\ref{eq-kmono}) follows trivially from the definition of $K^+$.

Now consider the general case, with no restriction on the sign of $\beta$. Thus, $\delta {\cal H}$ may be spacelike, timelike ($\beta<0$), or null ($\beta=0$); indeed, it may be spacelike at some portion of $\sigma(r)$ and timelike at another. One can still define the submanifold $X$ as the extension of $\Sigma^+(r)$ by $\delta {\cal H}$, as in Eq.~(\ref{eq-x}); see Fig.~\ref{fig-deltah}. Again, this extension cannot decrease the future of the set, nor its past domain of dependence,\footnote{The future of a set is defined for arbitrary sets. The domain of dependence is usually defined only for certain sets, for example for closed achronal sets in Ref.~\cite{Wald}. Here we extend the usual definition to the more general set $X$: $p\in D^-(X)$ iff every future-inextendible causal curve through $p$ intersects $X$. This is useful for our purposes; however, we caution that certain theorems involving $D^\pm$ need not hold with this broader definition.} as described in Eq.~(\ref{eq-idsuper}).

However, $X$ need not be achronal and hence, it need not lie on any Cauchy surface. In this case, we consider a new Cauchy surface that contains $\sigma(r+\delta r)$. Because $\alpha<0$, this surface can be chosen so that $\Sigma^+(r+\delta r)$ is nowhere to the future of $X$; see Fig.~\ref{fig-deltah}. Since $X$ and $\Sigma^+(r+\delta r)$ share the same boundary $\sigma(r+\delta r)$, $\alpha>0$ then implies that $X$ is entirely in the future of $\Sigma^+(r+\delta r)$: 
\begin{equation}
X\subset I^+(\Sigma^+(r+\delta r))
\label{eq-xfut}
\end{equation}
Moreover, the set $X$ together with $\bar\Sigma^+(r+\delta r)$ forms a ``box'' that bounds an open spacetime region $Y$, such that 
\begin{figure}[t]
\begin{center}
\includegraphics[width=3.0in]{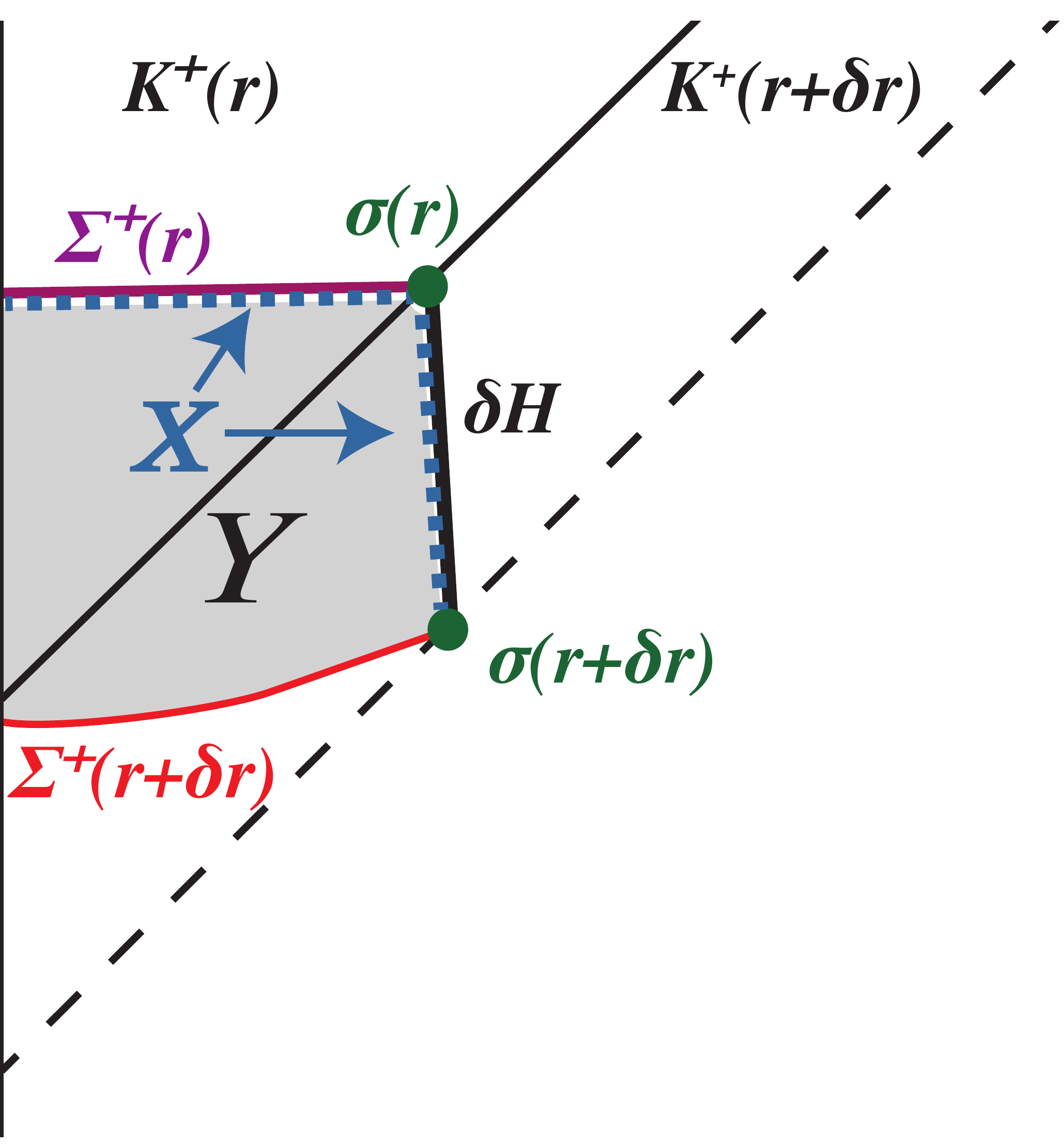}
\caption{Proof that $K^+(r)$ grows monotonically under inclusion, for any foliation $\sigma(r)$ of a hypersurface $\cal H$ with $\alpha<0$. See the main text for details and definitions.}
\label{fig-deltah}
\end{center}
\end{figure}
\begin{equation}
Y\subset  I^+(\Sigma^+(r+\delta r))~.
\label{eq-ysubi}
\end{equation}
All future-directed timelike curves that pass through $\Sigma^+(r+\delta r)$ enter $Y$ and then can exit $Y$ only through $X$. Hence $D^-(X) \subset Y \cup D^-(\Sigma^+(r+\delta r))$. Since $\alpha<0$, for all points outside of $Y \cup D^-(\Sigma^+(r+\delta r))$ there exist future-directed timelike curves that evade $X$. Hence equality holds:
\begin{equation}
D^-(X) = Y \cup D^-(\Sigma^+(r+\delta r))~.
\label{eq-dmx}
\end{equation}
To obtain the infinitesimal inclusion relation, 
\begin{equation}
K^+(r+\delta r)\supset K^+(r)~,
\label{eq-minimono}
\end{equation}
by Eq.~(\ref{eq-idsuper}) it suffices to show that $K^+(r+\delta r)\supset I^+(X)\cup D^-(X)$. Indeed if $p\in I^+(X)$ by Eq.~(\ref{eq-xfut}) $p\in I^+(\Sigma^+(r+\delta r))\subset K^+(r+\delta r)$. And if $p\in D^-(X)$ then by Eqs.~(\ref{eq-dmx}) and (\ref{eq-ysubi}) we again have $p\in  K^+(r+\delta r)$.

To obtain the stricter relation
\begin{equation}
K^+(r+\delta r)\supset \bar K^+(r)~,
\label{eq-minimonos}
\end{equation}
we note that $\sigma(r)\subset X$; hence by Eq.~(\ref{eq-xfut}), for every point $p\in\sigma(r)$ there exists a timelike curve from $\Sigma^+(r+\delta r)$ to $p$. This curve can be continued along the null generator of $N^+(r)$ starting at $p$ to a point $q\in N^+(r)$, and then slightly deformed into a timelike curve connecting $p$ to $q$. By Lemma~\ref{lem-secondN}, every point in $N^+(r)$ lies on a generator starting at $\sigma(r)$. Hence, $N^+(r)\subset K^+(r+\delta r)$. A similar argument yields $N^-(r)\subset K^+(r+\delta r)$. Since $N(r)=N^+(r)\cup N^-(r)$ and $\bar K^+(r) = K^+(r)\cup N(r)$, Eq.~(\ref{eq-minimonos}) follows. 

To extend Eq.~(\ref{eq-minimonos}) to Eq.~(\ref{eq-kmono}), one may iterate the above infinitesimal construction. The only way this could fail is if the iteration gets stuck because the steps $\delta r$ have to be taken ever smaller to keep Eq.~(\ref{eq-rdr}) satisfied. Suppose therefore that the iteration can only reach an open set $(r,r_*)$ but no leaves in the set $(r_*,r_2)$. But this contradicts the assumption that $\alpha<0$ at $r_*$.
\end{proof}

\section{Area Law}
\label{sec-arealaw}

In this section, we prove our main result: that the area of the holographic screen is monotonic. The most difficult part of this task is proving that $\alpha$ cannot change sign on $H$, Theorem~\ref{thm-alpha}).
%; thus, with Convention~\ref{conv-orient}, $\alpha<0$ everywhere on $H$. 
We then prove our Area Theorem~\ref{thm-area}.
\begin{figure}[h]
\centering
\includegraphics[height=0.35 \textwidth]{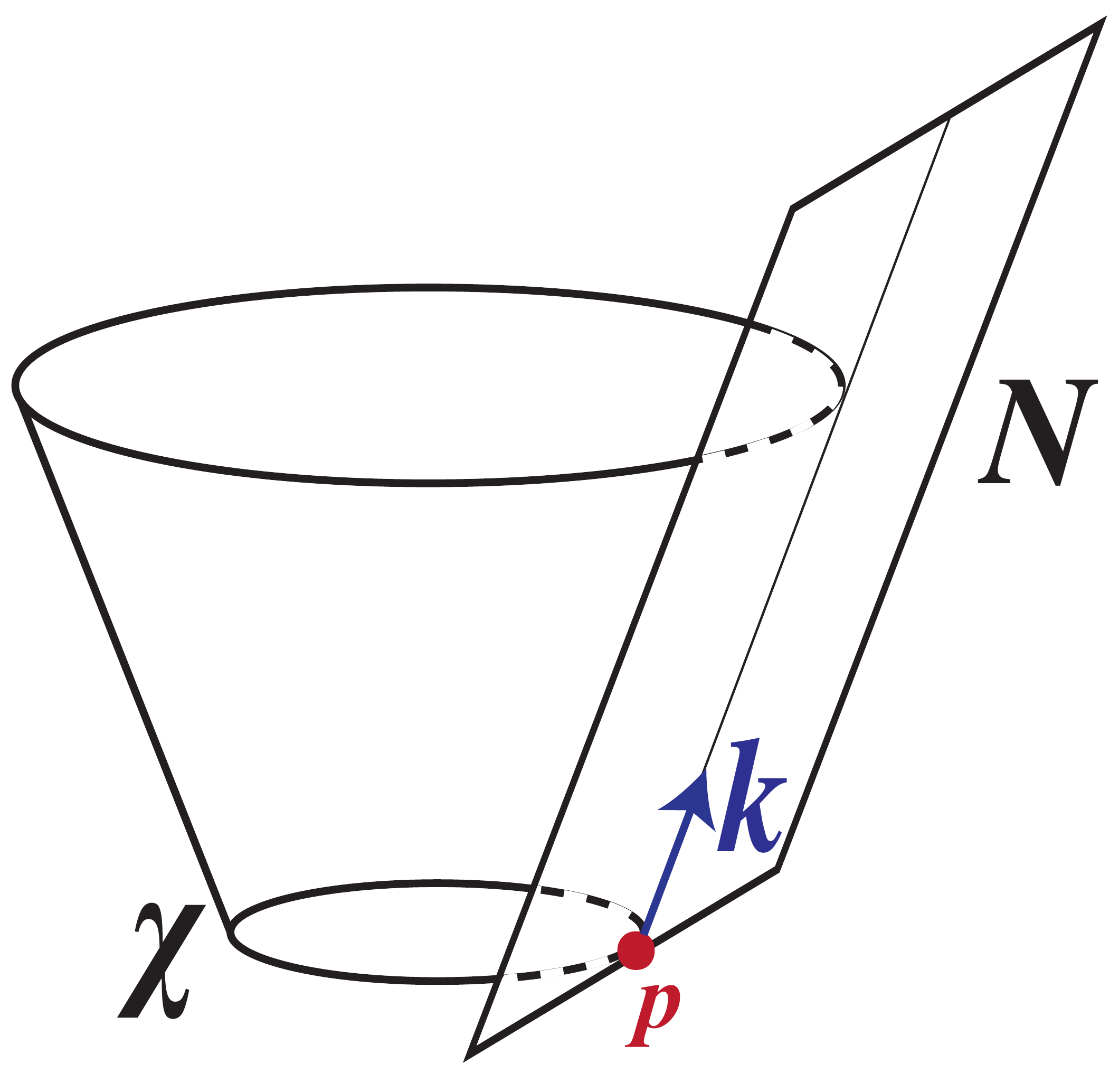}
\caption{An example illustrating Lemma~\ref{monotonicity}: in Minkowski space, the spatial sphere $\chi$ is tangent to the null plane $N$ at $p$ and lies outside the past of $N$ near $p$. It is easy to see that this implies that $\chi$ is a cross-section of a future light-cone that shares one null generator with $N$.  In this example it is obvious that $\chi$  expands faster than $N$ at $p$, as claimed in Lemma~\ref{monotonicity}.}
\label{fig-lemma}
\end{figure}
We begin by stating a useful Lemma.
\begin{lem}  \label{monotonicity}
Let $N$ be a null hypersurface and let $\chi$ be a spacelike surface tangent to $N$ at a point $p$. That is, we assume that one of the two future-directed null vectors orthogonal to $\chi$, $\kappa^a$, is also orthogonal to $N$ at $p$. We may normalize the (null) normal vector field to $N$ so that it coincides with $\kappa^a$ at $p$.  Let $\theta^{(\chi)}$ be the null expansion of the congruence orthogonal to $\chi$ in the $\kappa^a$ direction, and let $\theta^{(N)}$ be the null expansion of the generators of $N$. Then:
\begin{itemize} 
\item If there exists an open neighborhood $O(p)\cap \chi$ that lies entirely outside the past of $N$,\footnote{I.e., there exists no past directed causal curve from any point on $N$ to any point in $O(p)\cap \chi$.} then $\theta^{(\chi)}\geq \theta^{(N)}$ at $p$.
\item If there exists an open neighborhood $O(p)\cap \chi$ that lies entirely outside the future of $N$, then $\theta^{(\chi)}\leq \theta^{(N)}$ at $p$.
%\item If there exists an open neighborhood $O(p)\cap \chi$ that lies entirely on $N$, then $\theta^{(\chi)}= \theta^{(N)}$ at $p$. [THIS IS LEMMA B IN WALL; NOT USED HERE.]
\end{itemize} 
\end{lem}
\begin{proof}
See Lemma A in Ref.~\cite{Wal10QST}. Our Lemma is stronger but the proof is the same; so instead of reproducing it here, we offer Fig.~\ref{fig-lemma} to illustrate the result geometrically. It generalizes to null hypersurfaces an obvious relation in Riemannian space, between the extrinsic curvature scalars of two codimension 1 surfaces that are tangent at a point in a Riemannian space but do not cross near that point.
\end{proof}

% The converse of the above lemma is generally false, since $\theta$ is a trace. For example, consider a two-dimensional surface $\chi$ that touches a null plane $N$ in 3+1 dimensional Minkowski space on a point $p$. If the curvature components of $\chi$ in two orthogonal directions have opposite sign, $\sgn(\theta_1)\sng(\theta_2)=-1$, then $\chi$ curves to the future of $N$ in one direction and to the past in the orthogonal direction. Hence $\chi$ will cross $N$ at the tangent point, independently of the sign of $\theta=\theta_1+\theta_2$.

% However, under a strong additional assumption a converse does hold:
% \begin{lem}
% If the $D-2$ dimensional surface $\chi$ is tangent to the null hypersurface $N$ on a $D-3$ dimensional submanifold $P$, and if $\theta^{(\chi)}<\theta^{(N)}$ ($\theta^{(\chi)}>\theta^{(N)}$)for all $p\in P$, then $\chi\cap O(P)-P$ lies in the past (future) of $N$ for any sufficiently small open neighborhood of $P$. 
% \end{lem}
% \begin{proof}
% The set of tangent points $P$ is of codimension 1 in $\chi$ by assumption. Then at any point $p\in P$, the $D-3$ curvature radii along the directions tangent to $P$ are the same for $\theta$ and $N$. Thus only one curvature radius remains undetermined at any $p\in P$, and no cancellations are possible: if this curvature radius is smaller (greater) for $\chi$ than for $N$, then $\chi$ will bend to the past (future) of $N$ on both sides of $P$.
% \end{proof}

\begin{thm} Let $H$ be a regular future holographic screen with leaf-orthogonal tangent vector field $h^{a}=\alpha l^{a} + \beta k^{a}$, whose orientation is chosen so that $\alpha<0$ at the leaf $\sigma(0)$. Then $\alpha< 0$ everywhere on $H$.   \label{thm-alpha}
\end{thm}

\begin{proof}
By contradiction: suppose that $H$ contains a point with $\alpha\geq 0$. It immediately follows that the subset $H_+\subset H$ of points with $\alpha> 0$ is nonempty, since by Assumption~\ref{def-technical}.b, $\alpha=0$ can occur only as a transition between $\alpha<0$ and $\alpha>0$ regions. Let $\sigma(0)$ be the complete leaf that exists by Remark~\ref{rem-s0exists} and has $r=0$, $\alpha<0$, by Convention~\ref{conv-orient}. By continuity of $\alpha$, there exists an open neighborhood of $\sigma(0)$ where $\alpha<0$. 

\begin{figure}[t]
\includegraphics[height=0.3 \textwidth]{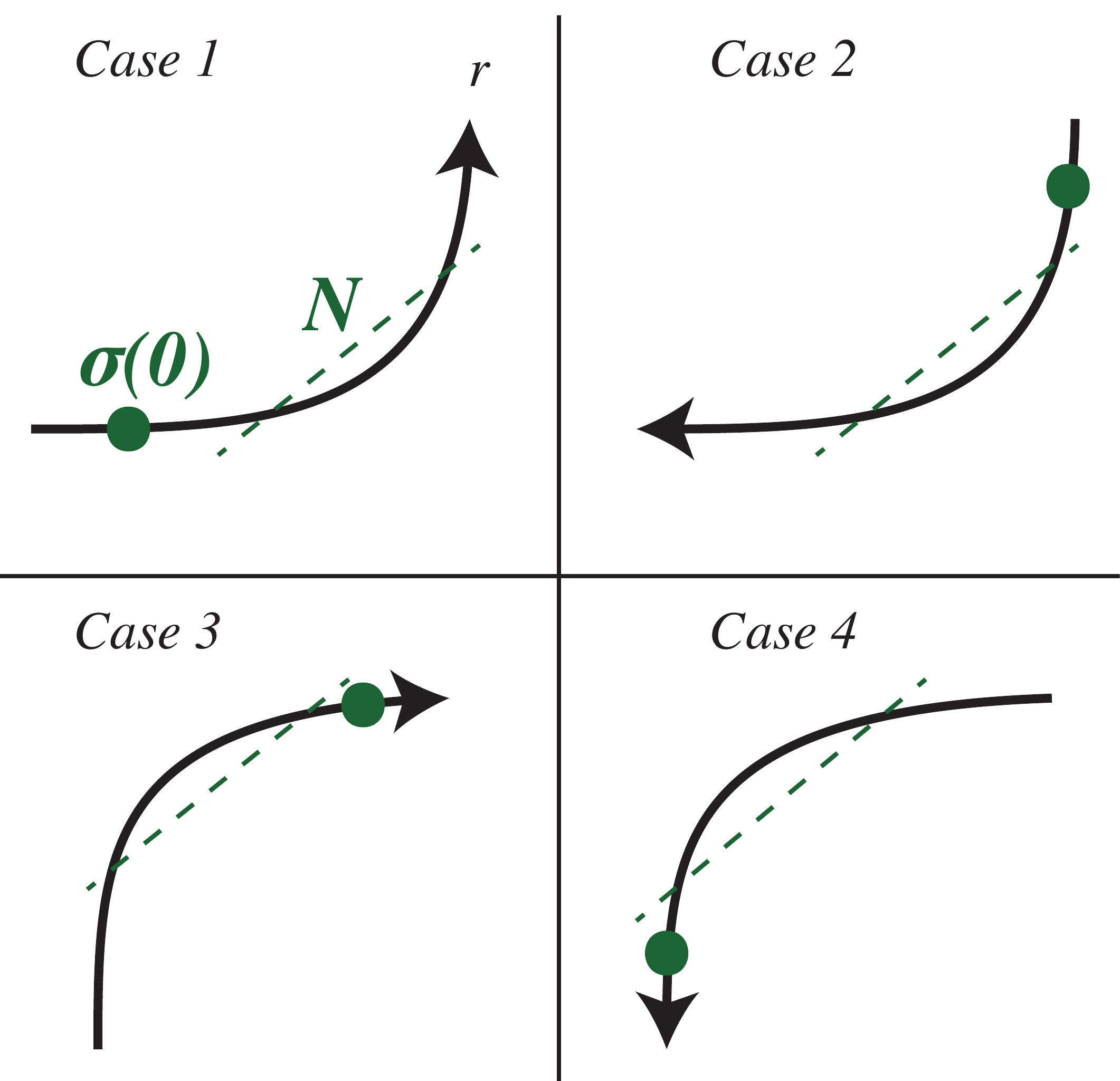}
\caption{The four types of spacelike-timelike transitions on a future holographic screen that would violate the monotonicity of the area, and which our proof in Sec.~\ref{sec-arealaw} will exclude. Near $\sigma(0)$, the area increases in the direction of the arrow. On the far side of the ``bend'' the area would decrease, in the same direction. There are other types of spacelike-timelike transitions which preserve area monotonicity under uniform flow; these do arise generically (see Fig.~\ref{fig-examplesPRD}a).}
\label{fig-sphsym}
\end{figure}
We first consider the case where $H_+$ has a component in the $r>0$ part of $H$ (cases 1 and 2 in Fig.~\ref{fig-sphsym}). Let $\sigma(1)$ be the ``last slice'' on which $\alpha\leq 0$, i.e., we use our freedom to rescale $r$ to set
\begin{equation}
1=\inf\{r:r>0,\sigma(r)\cap H_+ \neq \varnothing\}
\end{equation}
By the second generic condition~\ref{def-technical}.b, $\alpha<0$ for all leaves $\sigma(r)$ with $0<r<1$; hence by Theorem~\ref{thm-kmono} we have $K^-(0)\supset \bar K^-(1)$. Since the former set is open and the latter is closed, there exists an open neighborhood of $\bar K^-(1)$ that is contained in $K^-(0)$. Thus for sufficiently small $\epsilon$ we have
\begin{equation}
K^-(0)\supset K^-(1+\epsilon)~.
\label{eq-nocon}
\end{equation}

By continuity of $\alpha$, $\sigma(1)$ must contain at least one point with $\alpha=0$. Let $p$ denote this point; or, if there is more than one such point, let $p$ denote a connected component of the set of points with $\alpha=0$ on $\sigma(1)$. Since there is no point with $\alpha=\beta=0$, there exists an open neighborhood $O(p)\subset H$ in which $\beta$ has definite sign.  (Note that we do not assume that $\beta$ is of fixed sign for $0<r<1$.) 
\begin{figure*}[htb]
\centering
\subfigure[ ]{
\centering
\includegraphics[width=0.35 \textwidth]{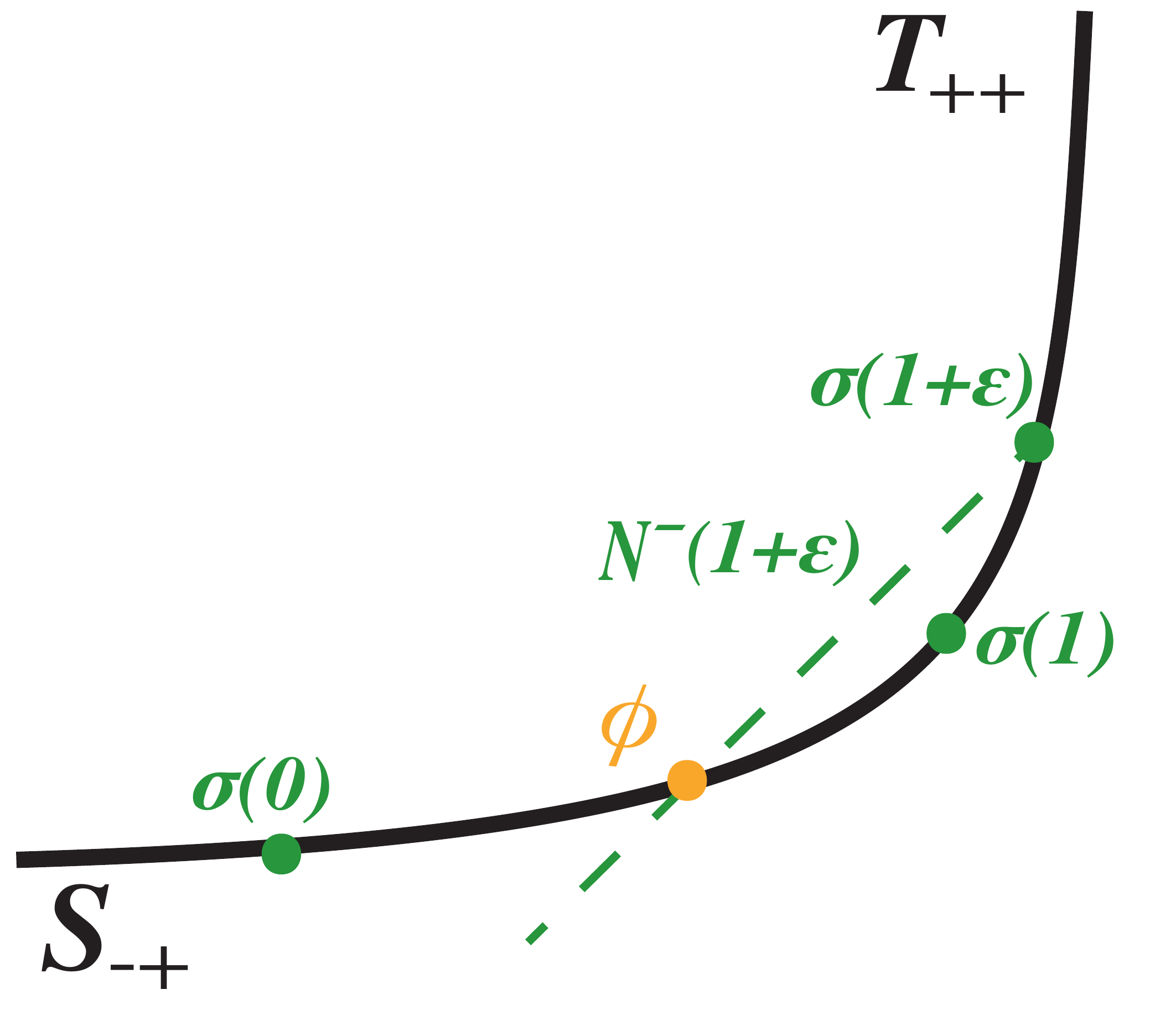}}
\qquad
\subfigure[]{
\centering \includegraphics[width=0.35\textwidth]{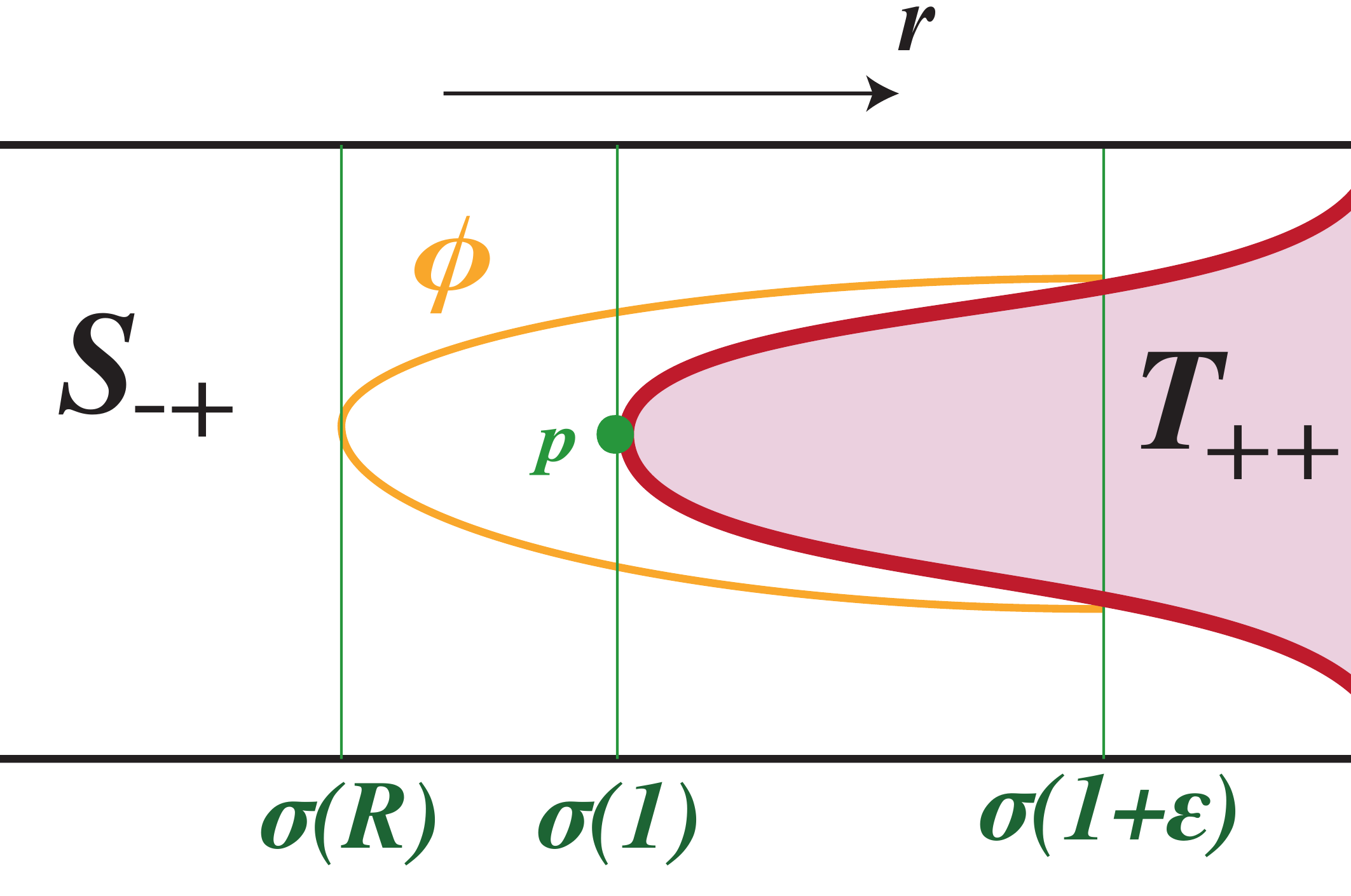}
}
%\hfill
\caption{A case 1 transition ($S_{-+}\to T_{++}$) is impossible. The proof crucially involves the intersection $\phi$ of a light-sheet $N$ originating just behind the assumed transition, with the region prior to the transition. For further details see the main text. (a) Spacetime diagram with two spatial directions suppressed. (b) Diagram of the holographic screen $H$ only, with only one spatial direction suppressed. Vertical lines correspond to leaves; top and bottom edge should be identified. Portions of different signature are indicated by shading and labels. In general, the transition boundary (thick red line) will not coincide with a leaf (thin black vertical lines).}
\label{fig-Case1LR}
\end{figure*}

\vskip .6cm
\noindent {\bf Case 1 \ } We now specialize further to the case where $\beta>0$ in $O(p)$, so that the assumed sign change from $\alpha<0$ to $\alpha>0$ corresponds to a transition of $h^a$ from spacelike-outward ($S_{-+}$) to timelike-future-directed ($T_{++}$). The following construction is illustrated in Fig. \ref{fig-Case1LR}.

Let $\sigma^+(1+\epsilon)$ be the set of points with $\alpha>0$ on the leaf $\sigma(1+\epsilon)$. If there is more than one connected component, we choose $\sigma^+(1+\epsilon)$ to be the component at least one of whose fibers intersects $p$. By choosing $\epsilon$ sufficiently small, we can ensure that $\sigma^+(1+\epsilon) \subset O(p)$. Let $\Gamma$ be the set of fibers that pass through $\sigma^+(1+\epsilon)$. 

Because $\alpha>0$, all fibers in $\Gamma$ enter $K^-(1+\epsilon)$ as we trace them back to smaller values of $r$. But $\sigma(0)$ is entirely outside of this set: by definition, $\sigma(0)\cap K^-(0)=\varnothing$, so Eq.~(\ref{eq-nocon}) implies $\sigma(0)\cap K^-(1+\epsilon)=\varnothing$. Hence, all fibers in $\Gamma$ also intersect $N(1+\epsilon)$, at some positive value of $r< 1+\epsilon$. Because $\beta>0$ in $O(p)$, this intersection will be with $N^-(1+\epsilon)$. By smoothness and the second generic assumption, the intersection will consist of one point per fiber. (Otherwise a fiber would coincide with a null generator of $N^{-}(1+\epsilon)$ in a closed interval.) 
%In this case we define $q(\gamma)$ to be point with smallest $r$ on this interval. 
The set of all such intersection points, one for each fiber in $\Gamma$, defines a surface $\phi$, and the fibers define a continuous, one-to-one map $\sigma^+(1+\epsilon)$ to $\phi$. Similarly, the closures of both sets, $\bar\sigma^+(1+\epsilon)$ and $\bar\phi$ are related by such a map. Note that these two sets share the same boundary at $r=1+\epsilon$.

Let $R$ be the minimum value of $r$ on the intersection: $R\equiv\inf\{r(q): q\in \bar\phi\}$. Since $\bar\sigma^+(1+\epsilon)$ is a closed subset of a compact set, it is compact; and by the fiber map, $\bar\phi$ is also compact. Therefore $R$ is attained on one or more points in $\bar\phi$. Let $Q$ be such a point. Since $R<1$ but $\dot\phi\subset \sigma(1+\epsilon)$, $Q\notin\dot\phi$, and hence $Q$ represents a local minimum of $r$. Hence the leaf $\sigma(R)$ is tangent to the null hypersurface $N^{-}(1+\epsilon)$ at $Q$.
% If $Q\in \partial\phi$ then $Q$ must lie on a fiber that coincides with a null generator of $N^{-}(1+\epsilon)$ connecting $\partial \sigma^+(1+\epsilon)$ to $\partial \phi$. Hence $h^a$ is orthogonal to $N^{-}(1+\epsilon)$ at $Q$, and so we find again that the leaf $\sigma(R)$ is tangent to the null hypersurface $N^{-}(1+\epsilon)$ at $Q$.

Since $Q$ achieves a global minimum of $r$ on $\bar\phi$, $\sigma(R)$ lies nowhere in the past of $N^{-}(1+\epsilon)$ in a sufficiently small open neighborhood of $Q$. For suppose there existed no such neighborhood. Then fibers arbitrarily close to the one containing $Q$ (and hence connected to $\sigma^+(1+\epsilon)$ would still be inside $K^-(1+\epsilon)$ at $R$. Hence we could find a value $r<R$ on $\phi$ by following such a fiber to smaller values of $r$ until it leaves $K^{-}(1+\epsilon)$. But this would contradict our construction of $Q$ as a point that attains the minimum value of $r$ on $\phi$.

Thus, Lemma~\ref{monotonicity} implies that $\theta_k^{\sigma(R)}\geq\theta_k^{N^{-}(1+\epsilon)}$ at $Q$. By the first generic assumption, the latter expansion is strictly positive, so we learn that $\theta_k^{\sigma(R)}>0$ at $Q$.  But this contradicts the defining property of holographic screens, that all leaves are marginally trapped ($\theta_k^{\sigma(r)}=0$ for all $r$). 

\vskip .6cm
\noindent {\bf Case 2 \ }  Next we consider the case where $\beta<0$ in the neighborhood of the assumed transition from $\alpha<0$ to $\alpha>0$ that begins at $r=1$ (see Fig.~\ref{fig-sphsym}). This corresponds to the appearance of a spacelike-inward-directed region within a timelike-past-directed region: $T_{--}\to S_{+-}$.

We note that the direct analogue of the above proof by contradiction fails: tracing back the generators from $\sigma^+(1+\epsilon)$ to $\sigma(0)$, one finds that they pass through $N^+(1+\epsilon)$, rather than $N^-(1+\epsilon)$. But $N^+$ has negative expansion by the first generic condition, whereas $N^-$ had positive expansion. There is no compensating sign change elsewhere in the argument; in particular, the tangent leaf $\sigma(R)$ with vanishing expansion again lies nowhere in the past of $N^+$ in a neighborhood of the tangent point $Q$. Thus no contradiction arises with Lemma~\ref{monotonicity}.

Instead, we show that every case 2 transition implies the existence of a case 1 transition at a {\em different}\/ point on $H$, under the reverse flow $r\to c-r$. Since we have already shown that case 1 transitions are impossible, this implies that case 2 transitions also cannot occur. 

Let us first illustrate this argument in the simple case where the transition occurs entirely on a single leaf: $\alpha<0$ for $0\leq r<1$, $\alpha=0$ at $r=1$, and $\alpha>0$ for $1<r\leq 2$. Under a reversal of the flow, $r\to 2-r$, $\alpha$ and $\beta$ both change sign. With this flow direction, the latter region now contains a leaf $\sigma(0)$ on which $\alpha<0$, and thus the starting point of our case 1 proof.  The putative sign change of $\alpha$ corresponds to a case 1 transition $S_{-+}\to T_{++}$. The case 1 proof by contradiction rules out this transition.

\begin{figure*}[htb]
\centering
\subfigure[ ]{
\centering
\includegraphics[height=0.25 \textwidth]{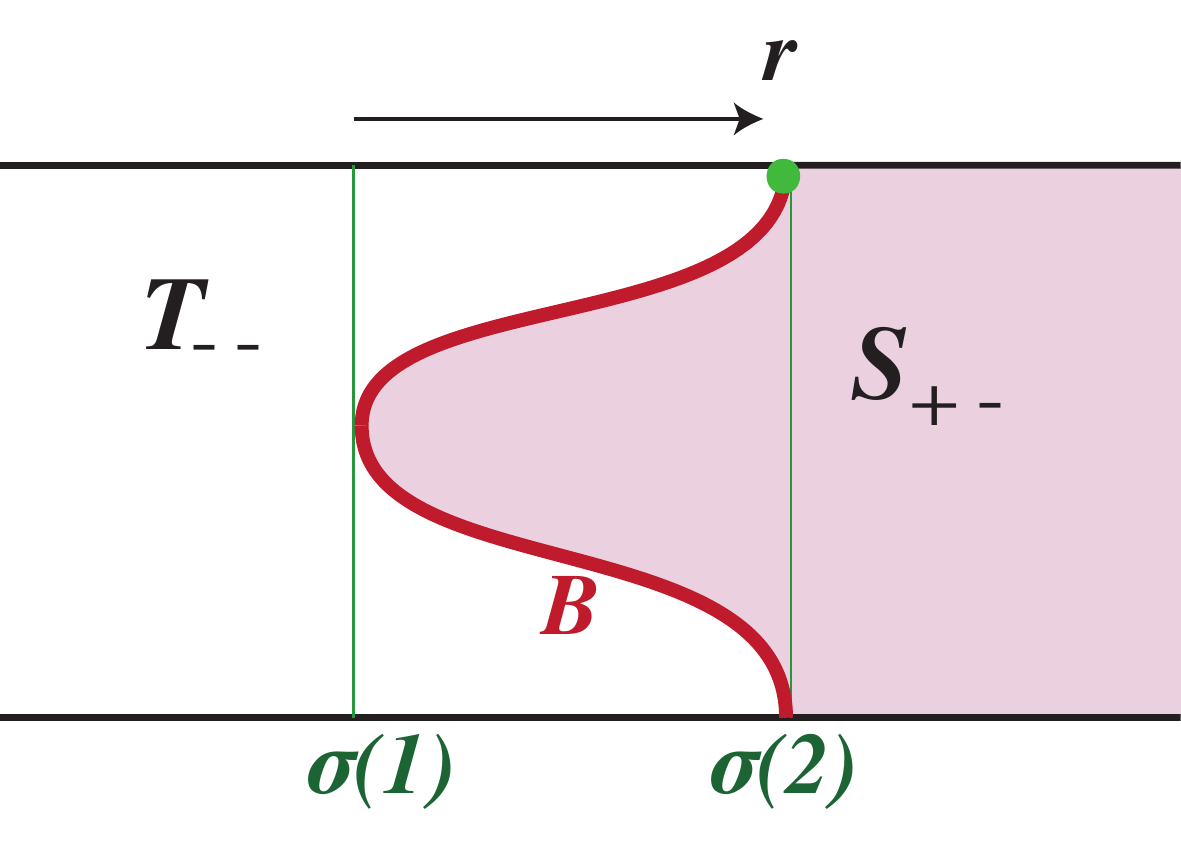}}
\qquad
\subfigure[]{
\centering \includegraphics[height=0.25\textwidth]{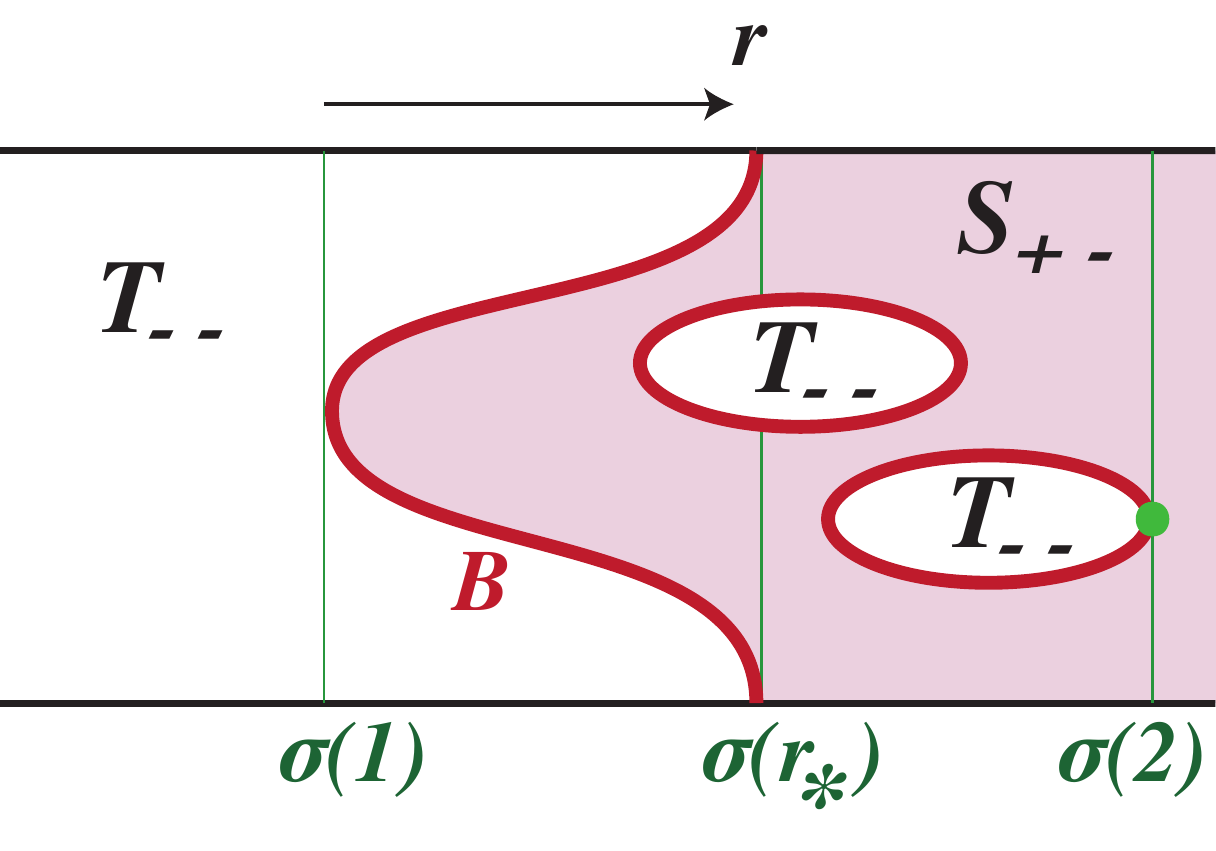}
}
%\hfill
\caption{A case 2 transition ($T_{--}\to S_{+-}$) is impossible. By assumption, the $\alpha>0$ region contains a complete leaf $\sigma(2+\epsilon)$. In the text we show that the complete-leaf region begins at some leaf $\sigma(2)$ where a $T_{--}\to S_{+-}$ boundary comes to an end: either the original one (a), or a different one containing a $T_{--}$ region with no complete leaf (b). The endpoint (green dot) becomes the starting point of a case 1 transition ($S_{-+}\to T_{++}$) under reversal of the flow direction; but this case has already been ruled out.}
\label{fig-case2proof}
\end{figure*}

In general, the case 2 transition need not occur on a single leaf, so we shall assume for contradiction only that $\alpha$ first becomes positive at some point on or subset of $\sigma(1)$, as in the case 1 proof, and that $\beta<0$ in a neighborhood of this set. Let $\tilde H_+$ denote the connected region with $\alpha>0$ that begins at this transition. Since the transition is $T_{--}\to S_{+-}$, $\tilde H_+$ contains some spacelike points; and hence by Def.~\ref{def-technical}.c, $\tilde H_+$ contains a complete leaf with $\alpha>0$. We use our freedom to rescale $r$ to set
\begin{equation}
2=\inf\{r:r>0,\sigma(r)\subset \tilde H_+\}
\label{eq-leaf2}
\end{equation}
By the second generic assumption, Def.~\ref{def-technical}.b, this choice implies the existence of an open interval $(2,2+\epsilon)$ such that every leaf in this interval is a complete leaf with $\alpha>0$. Let us call this intermediate result (*); see Fig.~\ref{fig-case2proof} which also illustrates the remaining arguments.

We now consider the boundary $B$ that separates the $\alpha<0$ from the $\alpha>0$ region, i.e., the connected set of points with $\alpha=0$ that begins at $r=1$. Because $\alpha$ and $\beta$ cannot simultaneously vanish, we have $\beta<0$ in an open neighborhood of all of $B$. Thus, $B$ separates a $T_{--}$ region at smaller $r$ from a $S_{+-}$ region at larger $r$. We note that $B$ must intersect every fiber, or else $H_+$ would not contain a complete leaf. Moreover, $B$ must end at some $r_*\leq 2$, or else there would be points with $\alpha<0$ in the interval $(2,2+\epsilon)$, in contradiction with (*). 

If $r_*=2$ then under the reverse flow starting from the complete leaf at $r=2+\epsilon$ there is a case 1 transition at $r=2$ from $S_{-+}$ to $T_{++}$, and we are done. This is shown in Fig.~\ref{fig-case2proof}a.  

The only remaining possibility is that $B$ ends at some $r_*\in (1,2)$; this is shown in Fig.~\ref{fig-case2proof}b. Then every leaf with $r\in (r_*,2)$ must contain points with $\alpha<0$, or else there would be a complete leaf with $\alpha>0$ at some $r<2$, in contradiction with Eq.~(\ref{eq-leaf2}). Therefore each leaf with $r\in (r_*,2)$ must intersect one or more $\alpha<0$ regions $\tilde H_-^{(i)}$ that are disconnected from the $T_{--}$ region bounded by $B$. None of these regions $\tilde H_-^{(i)}$ can contain a complete $\alpha<0$ leaf, because this would imply that $\tilde H_+$ does not contain a complete $\alpha>0$ leaf. From Def.~\ref{def-technical}.c it follows that each region $\tilde H_-^{(i)}$ is everywhere timelike, i.e., of type $T_{--}$. But this implies that a $T_{--}$ region ends at $r=2$ where $\alpha$ becomes positive. Moreover, the $S_{+-}$ region in which the $T_{--}$ region ends has complete leaves in some open interval $(2,2+\epsilon)$ by our result (*). Thus we find again that under the reverse flow starting from the complete leaf at $r=2+\epsilon$ there is a case 1 transition at $r=2$ from $S_{-+}$ to $T_{++}$.

We have thus established that a case 2 transition at $r=1$ implies a case 1 transition at the same or a larger value of $r$, after reversal of the direction of flow. Since case 1 transitions are impossible, we conclude that case 2 transitions are also impossible.
\begin{figure*}[ht]

\includegraphics[width= 0.85\textwidth]{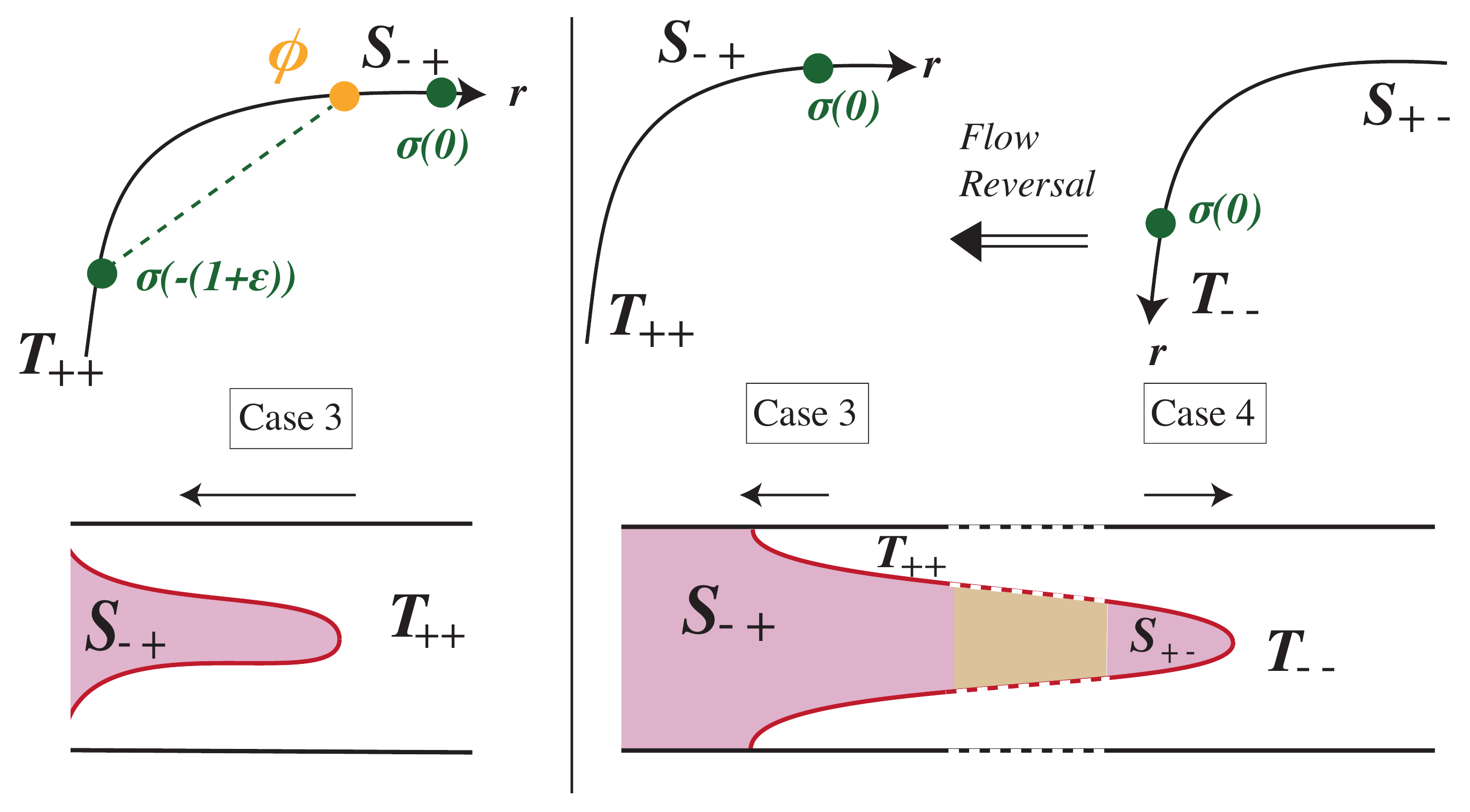}
\caption{(a) Case 3 is ruled out analogously to case 1, by contradiction.  (b) Case 4 is analogous to case 2: the transition is impossible because it would imply a case 3 transition elsewhere on $H$, under reversal of the flow direction.}
\label{fig-cases}
\end{figure*}

\vskip .6cm
\noindent {\bf Cases 3 and 4} 
Our consideration of cases 1 and 2 has ruled out the possibility of points with $\alpha>0$ at any $r>0$. (Recall that $r=0$ corresponds to a complete leaf with $\alpha<0$.) We must now also rule out the possibility that $\alpha$ might be positive in the region $r<0$; this corresponds to cases 3 and 4 in Fig.~\ref{fig-sphsym}. Again, assume for contradiction that such a transition occurs, and focus on the transition nearest to $r=0$. We may rescale $r$ so that this transition ends at $r=-1$. That is, $\alpha<0$ for all $r\in (-1,0)$, but all leaves in some interval $(-(1+\epsilon),-1)$ contain points with $\alpha>0$. Again, a further case distinction arises depending on the sign of $\beta$ at this transition.
 
The proof of case 3 (Fig.~\ref{fig-cases}a), where $\beta>0$ at the transition, proceeds exactly analogous to that of case 1. Fibers that connect the offending region to $r=0$ must cross the null hypersurface $N^+(-(1+\epsilon))$, implying the existence of a leaf $\sigma(R)$, $-1<R<0$ that is tangent to $N^+(-(1+\epsilon))$ and nowhere to the future of $N^+(-(1+\epsilon))$. But $N^+$ contracts at this tangent point whereas $\sigma(R)$ has vanishing expansion, in contradiction with the second part of Lemma~\ref{monotonicity}.

The proof of case 4 (Fig.~\ref{fig-cases}b) proceeds analogous to that of case 2, by showing that a case 4 transition at $r=-1$ implies the existence of a transition at some $r\leq -1$ that is recognized as a case 3 transition after reversal of the flow direction, and hence ruled out.

\end{proof}

\begin{figure*}[ht]
\includegraphics[width=6.5in]{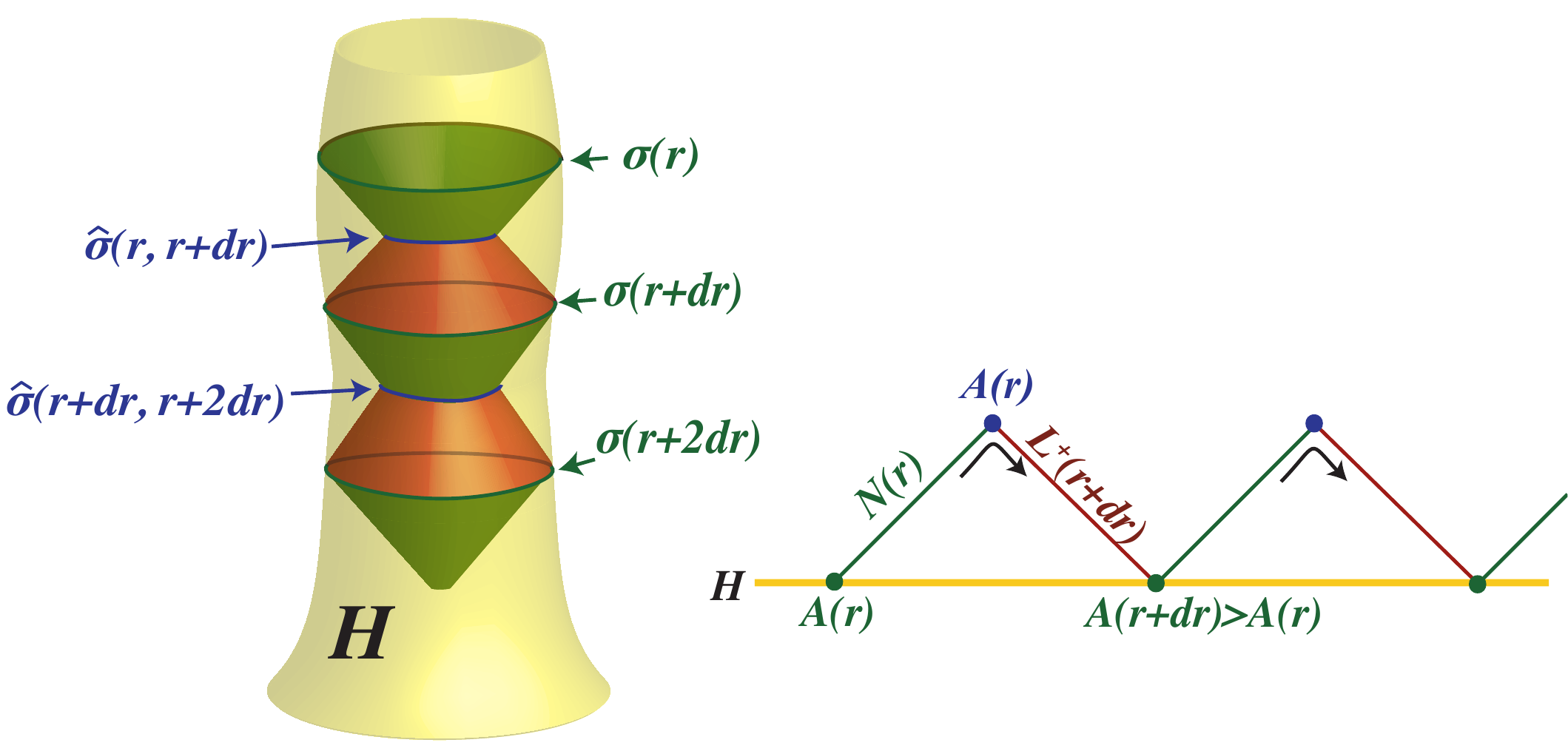}
\caption{Proof that $\alpha<0$ implies an area law, by a zig-zag construction. For any pair of infinitesimally nearby leaves, consider the intersection $\hat\sigma$ of the light-sheet pair $N(r)$ tangent to the marginal direction $k^a$, with the lightsheet $L^+(r+\delta r)$ generated by $l^a$. To first order in $\delta r$, the area is constant as we leave $H$ from $\sigma(r)$ along $N(r)$ to $\hat \sigma$, and the area increases as we follow $L^+$ from $\hat\sigma$ back onto $H$. This construction assumes only that $\alpha<0$, so that evolution is in the $-l^a$ direction; it does not require $H$ to have uniform signature near $\sigma(r)$.  (a) Timelike case, with one spatial dimension suppressed and relevant 2-surfaces labeled. (b) Spacelike case, with all spatial dimensions suppressed. In this plot, we choose to label the relevant null surfaces; the two-surfaces are labeled by their area not by their name.}
\label{fig-zigzag}
\end{figure*}
We now state and prove the area law.
\begin{thm} \label{thm-area}
The area of the leaves of any regular future holographic screen $H$ increases strictly monotonically:
\begin{equation}
\frac{dA}{dr}>0~.
\label{eq-area}
\end{equation}
\end{thm}
\begin{proof} 
By Theorem~\ref{thm-alpha}, $\alpha<0$ everywhere on $H$. In regions where $\beta$ is of definite sign, the result would then follow from the analysis of Hayward~\cite{Hay93} (using a 2+2 lightlike formalism) or that of Ashtekar and Krishnan~\cite{AshKri02} who used a standard 3+1 decomposition. It should be straightforward to generalize their proofs to the case where $\beta$ may not have definite sign on some or all leaves. However, since this would necessitate the introduction of additional formalism, we will give here a simple, geometrically intuitive proof. Our construction is shown in Fig.~\ref{fig-zigzag}. 

Consider two infinitesimally nearby leaves at $r$ and $r+dr$, $dr>0$. Construct the null hypersurface $N(r)$ in a neighborhood of $\sigma(r)$. Also, construct the null hypersurface $L^+(r+dr)$ generated by the future-directed null geodesics with tangent vector $l^a$, in a neighborhood of $\sigma(r+dr)$. By Theorem~\ref{thm-alpha}, for sufficiently small $dr$ these null hypersurfaces intersect on a two-dimensional surface $\hat\sigma(r,r+dr)$, such that every generator of each congruence lies on a unique point in $\hat\sigma(r,r+dr)$. 

Note that in regions where $H$ is spacelike, $\beta>0$, the intersection will lie in $N^+(r)$; if $H$ is timelike, $\beta<0$, the intersection will lie in $N^-(r)$; but this makes no difference to the remainder of the argument. Crucially, Theorem~\ref{thm-alpha} guarantees that the intersection always lies in $L^+(r+dr)$, and never on $L^-(r+dr)$, the null hypersurface generated by the past-directed null geodesics with tangent vector $-l^a$.
We now exploit the defining property of $H$, that each leaf is marginally trapped ($\theta^{\sigma(r)}_k=0$). This implies
\begin{eqnarray} 
A[\hat\sigma]-A[\sigma(r)] & = & O(dr^2)~;\\
A[\sigma(r+dr)]-A[\hat\sigma] & = & O(dr)>0~.
\end{eqnarray} 
Hence, the area increases linearly in $dr$ between any two nearby leaves $\sigma(r)$, $\sigma(r+dr)$. This implies that the area increases strictly monotonically with $r$.

\end{proof}

\begin{cor}\label{cor-quant}
The above construction implies, more specifically, that the area of leaves increases at the rate
\begin{equation}
\frac{dA}{dr} = \int_{\sigma(r)} \sqrt{h^{\sigma(r)}}~ \alpha \theta^{\sigma(r)}_l~.
\label{eq-specarea}
\end{equation}
where $h_{ab}^{\sigma(r)}$ is the induced metric on the leaf $\sigma(r)$ and $h^{\sigma(r)}$ is its determinant. Note that the integrand is positive definite since $\alpha<0$ and all leaves are marginally trapped; in this sense the area theorem is local. However, the theorem applies to complete leaves only, not to arbitrary deformations of leaves.
\end{cor}

\begin{cor}\label{cor-past}
For past holographic screens, we recall the contrasting convention that $\alpha>0$ on $\sigma(0)$. The above arguments then establish that $\alpha>0$ everywhere on $H$. Eqs.~(\ref{eq-area}) and (\ref{eq-specarea}) hold as an area theorem. 
\end{cor}

\begin{rem}\label{rem-arrow}
We note that the area increases in the outside or future direction along a past holographic screen. With an interpretation of area as entropy, the holographic screens of an expanding universe thus have a standard arrow of time. 
\end{rem}
\begin{rem}\label{rem-backarrow}
By contrast, the area increases in the outside or {\em past} direction along a future holographic screen. Thus,  {\em the arrow of time runs backwards on the holographic screens inside black holes, and near a big crunch}. Perhaps this intriguing result is related to the difficulty of reconciling unitary quantum mechanics with the equivalence principle~\cite{Haw76,AMPS,Bou12c,AMPSS,Bou13,MarPol13,Bou13a,Bou13b}.
\end{rem}

We close with a final theorem that establishes the uniqueness of the foliation of $H$:
\begin{thm}
Let $H$ be a regular future holographic screen with foliation $\{\sigma(r)\}$. Every marginally trapped surface $s\subset H$ is one of the leaves $\sigma(r)$. 
\end{thm}
\begin{proof}
By contradiction: suppose that $s$ is marginally trapped and distinct from any $\sigma(r)$.  Thus $s$ intersects the original foliation in a nontrivial closed interval $[r_1,r_2]$ and is tangent to $\sigma(r_1)$ and $\sigma(r_2)$. The $\theta=0$ null vector field orthogonal to $s$ must coincide with $k^a$ at the tangent point with $\sigma(r_2)$. Since $r_1<r_2$, Theorems~\ref{thm-alpha} and \ref{thm-kmono} imply that $N(r_2)$ does not everywhere coincide with the null hypersurface orthogonal to $s$ with tangent vector $k^a$ at $\sigma(2)$. Lemma B in Ref.~\cite{Wal10QST} then implies that $\theta_{k^{(s)}}^{(s)}\neq 0$ somewhere on $s$, in contradiction with the assumption that $s$ is marginally trapped.
\end{proof}

\vskip .3cm
\indent {\bf Acknowledgments} 
It is a pleasure to thank S.~Fischetti, D.~Harlow, G.~Horowitz, W.~Kelly, S.~Leichen\-auer, D.~Marolf, R.~Wald, and A.~Wall for discussions and correspondence. NE thanks the Berkeley Center for Theoretical Physics and the UC Berkeley Physics Department for their hospitality. The work of RB is supported in part by the Berkeley Center for Theoretical Physics, by the National Science Foundation (award numbers 1214644 and 1316783), by fqxi grant RFP3-1323, and by the US Department of Energy under Contract DE-AC02-05CH11231. The work of NE is supported in part by the US NSF Graduate Research Fellowship under Grant No. DGE-1144085 and by NSF Grant No. PHY12-05500.

\bibliographystyle{utcaps}
\bibliography{all}
\end{document}